\documentclass[11pt, oneside]{article}  	
\usepackage{geometry}        		
\usepackage{graphicx}				
\usepackage{amsthm}							
\usepackage{bbm}	
\usepackage{authblk}
\usepackage{amssymb}
\usepackage{amsmath}
\usepackage{booktabs}
\usepackage{float}
\usepackage{makecell}
\usepackage{tikz}
\usetikzlibrary{er,positioning}
\usepackage[super,comma,sort&compress]{natbib}

\makeatletter
\newcommand*{\citenst}[2][]{%
 \begingroup
 \let\NAT@mbox=\mbox
 \let\@cite\NAT@citenum
 \let\NAT@space\NAT@spacechar
 \let\NAT@super@kern\relax
 \renewcommand\NAT@open{[}%
 \renewcommand\NAT@close{]}%
 \citet[#1]{#2}%
 \endgroup
}
\makeatother

\usepackage{color}
\definecolor{darkred}{RGB}{100,0,0}
\definecolor{darkgreen}{RGB}{0,100,0}
\definecolor{darkblue}{RGB}{0,0,150}

\usepackage{hyperref}
\usepackage{mathtools}
\hypersetup{colorlinks=true, linkcolor=darkred, citecolor=darkgreen, urlcolor=darkblue}
\usepackage{url}

\newtheorem{thm}{Theorem}
\newtheorem{prp}{Proposition}
\newtheorem{lem}{Lemma}

\newtheorem{ass}{Condition}

\theoremstyle{remark}

\newtheorem{rem}{Remark}

\def\beq{\begin{equation}} 
\def\eeq{\end{equation}}
\def\beqn{\begin{eqnarray*}}
\def\eeqn{\end{eqnarray*}}
\def\Bitem{\begin{itemize}\setlength{\itemsep}{.2in}}
\def\bitem{\begin{itemize}\setlength{\itemsep}{.05in}}
\def\eitem{\end{itemize}}
\def\Benum{\begin{enumerate}\setlength{\itemsep}{.2in}}
\def\benum{\begin{enumerate}\setlength{\itemsep}{.05in}}
\def\eenum{\end{enumerate}}
\def\bmult{\begin{multline*}}
\def\emult{\end{multline*}}
\def\bcenter{\begin{center}}
\def\ecenter{\end{center}}
\def\bframe{\begin{frame}}
\def\eframe{\end{frame}}

\def\cF{\mathcal{F}}

\def\cK{\mathcal{K}}

\def\cP{\mathcal{P}}

\def\be{\mathbf{e}}

\def\bo{\mathbf{o}}

\newcommand{\bbeta}{{\boldsymbol\beta}}

\def\bbA{\mathbb{A}}

\def\bbR{\mathbb{R}}

\newcommand{\E}{\operatorname{\mathbb{E}}}
\renewcommand{\P}{\operatorname{\mathbb{P}}}
\newcommand{\Var}{\operatorname{Var}}

\def\1{\mathbbm{1}}

\def\xo{X_o}
\def\xi{X_I}
\def\xe{X_e}
\def\ao{\alpha_o}
\def\ai{\alpha_I}
\def\bo{\beta_o}
\def\be{\beta_e}

\def\xu{{X_u}}
\def\xoi{X_{oi}}
\def\xii{X_{Ii}}
\def\xei{X_{ei}}
\def\bbeta{\beta}

\def\szero{s^{(0)}}
\def\sone{s^{(1)}}

\def\Nit{N_i(t)}

\def\Yit{Y_i(t)}

\def\Mit{M_i(t)}

\def\rit{r_i(t)}

\def\Gt{G(t)}
\def\hGt{\hat{G}(t)}
\def\Gtit{G(T^*_i \wedge t)}
\def\hGtit{\hat{G}(T^*_i \wedge t)}

\def\hwit{\hat{w}_i(t)}
\def\twit{\tilde{w}_i(t)}
\def\Mi1t{M_i^1(t)}
\def\sMi1t{M_i^{1*}(t)}

\def\tX{\tilde{X}}

\definecolor{purple}{rgb}{0.4,.1,.9}

\newcommand{\eqn}{\begin{eqnarray}}
\newcommand{\ee}{\end{eqnarray}}
\newcommand{\eqnn}{\begin{eqnarray*}}
\newcommand{\een}{\end{eqnarray*}}

\title{\textbf{Two-Stage Residual Inclusion for Survival Data and Competing Risks - An Instrumental Variable Approach with Application to SEER-Medicare Linked Data}}
\author[1]{Andrew Ying}
\author[1,2]{Ronghui Xu}
\author[3]{James Murphy}
\affil[1]{Department of Mathematics}
\affil[2]{Department of Family Medicine and Public Health}
\affil[3]{Department of Radiation Medicine and Applied Sciences\\ University of California, San Diego, CA 92093, U.S.A}
\providecommand{\keywords}[1]{\textbf{\textit{Keywords:}} #1}
\date{}

\begin{document}
\maketitle

\begin{abstract}

Instrumental variable is an essential tool for addressing unmeasured confounding in observational studies. Two stage predictor substitution (2SPS) estimator and two stage residual inclusion(2SRI) are two commonly used approaches in applying instrumental variables. Recently 2SPS was studied under the additive hazards model in the presence of competing risks of time-to-events data, where linearity was assumed for the relationship between the treatment and the instrument variable. This assumption may not be the most appropriate when we have binary treatments. In this paper, we consider the 2SRI estimator under the additive hazards model for general survival data and in the presence of competing risks, which allows generalized linear models for the relation between the treatment and the instrumental variable. We derive the asymptotic properties including a closed-form asymptotic variance estimate for the 2SRI estimator. We carry out numerical studies in finite samples, and apply our methodology to  the linked Surveillance, Epidemiology and End Results (SEER) - Medicare database comparing radical prostatectomy versus conservative treatment in early-stage prostate cancer patients.

\end{abstract}

\keywords{Additive hazards model; Asymptotic variance; Causal inference; Endogeneity; Unobserved confounding.} 

\section{Introduction}

Interpreting the causal meaning of a treatment or exposure is straightforward under the randomized trials, because the randomization guarantees that there are no confounders for the exposure or treatment of interest. 
However, randomized experiment is not always feasible in practice. 
In observational studies issues such as endogeneity or potential confounding will arise. 
 Instrumental variable (IV) is a useful method in some of these settings \citep{angrist1996identification}, including when we may have unmeasured confounders. 
 It finds increasing application in research on health care practices (see for example
 \citet{hogan:lanc, stukel2007analysis}) including comparative effectiveness studies, and also in genetic studies where certain genes are used as IV for Mendelian randomization
\cite{glymour2012credible, lawlor2008mendelian, zhang2015assessing}. 
 
Figure \ref{fig:causal} illustrates a typical setting where IV methods can be applied. 
It is called a causal directed acyclic graph, where the nodes represent random variables, arrows represent direct causal effect, in such a way so that the common cause of any two nodes is included in the graph. In Figure \ref{fig:causal} $\xe$ is the exposure we are interested in testing the causal effect on the survival time outcome $T$, and we use a dashed line to represent the uncertainty of causation. In addition, $\xo$ is the observed confounder, $\xu$ the unobserved confounder. $\xi$ is the instrument variable which has to satisfy the following three conditions:
1) $\xi$ is associated with $\xe$,
2) $\xi$ doesn't affect $T$ except through its potential effect on $\xe$,
and 3) $\xi$ and $\xu$ do not share causes.

There are two commonly used IV approaches: two-stage predictor substitution (2SPS), and two-stage residual inclusion (2SRI). 
For survival outcomes \citet{tchetgen2015instrumental} 
considered these two methods under the additive hazards model, and gave conditions under which the causal parameters of interest can be correctly estimated. 
We note that while the Cox proportional hazards model 
has been more widely used in practice for survival data, an important appeal of additive hazards models is that unlike proportional hazards, a hazards difference is a collapsible effect measure.\cite{tchetgen2015instrumental, mart:vans}
This means that when there is an unobserved (exogenous) covariate for the survival outcome, integrating out this covariate under the additive hazards model still gives a model satisfying the additive hazards assumption. But this has long been known not to be the case for the proportional hazards model.\cite{lanc:nick, gail:etal, stru:kalb, bret:hub, ande:flem, ford:etal, xu:oq}
In 2SPS under the additive hazards model the exposure is substituted by its prediction given the IV, and included in the model as a covariate, together with possibly other observed confounders. \citet{li2015instrumental} investigated its large sample properties under the additive hazards model. \citet{zheng2017instrumental} considered 2SPS under competing risks. 

The 2SRI is different from 2SPS in that instead of the predicted exposure, the residual from regressing the exposure on the IV(i.e.~first stage) is included as an extra term in the additive hazards model, in addition to the exposure and possibly other observed confounders. This was considered more suitable for binary or discrete exposures \cite{terza2008two, tchetgen2015instrumental}, when generalized linear models (GLM) for example are used in the first stage. Note that linear regression is used in the first stage of 2SPS, and therefore it is more suitable for continuous exposures. 
Very recently \citet{jiang:etal:2018} studied  the 2SRI estimator with linear regression in the first stage for survival data. 
Our goal in this paper is to study the 2SRI estimator with GLM in the first stage, and also to develop the methodology for competing risks data. 

This work was motivated by the desire to conduct comparative effectiveness research in large observational databases. In the field of oncology we lack gold standard randomized clinical trials in many clinical scenarios to optimally inform clinical decision making. With the lack of randomized trials investigators turn to comparative effectiveness with large observational data sets such as the linked (Surveillance, Epidemiology, and End Results) SEER-Medicare database. These databases include information on the specifics of cancer, staging, treatment, patient comorbidity, as well as information on long-term outcomes including toxicity and survival. Despite this wealth of information, these databases do not contain information on unmeasured confounding factors such as patient weight, smoking status, diet, exercise, patient compliance with treatment, and patient performance status. These unmeasured confounders can substantially influence outcomes (in particular survival), adding bias to comparative effectiveness research using observational data. 
Recently \citet{hadley2010comparative} compared aggressive (radical prostatectomy) versus conservative treatments of prostate cancer with SEER-Medicare linked data. The question of radical prostatectomy versus conservative treatment has been addressed in randomized clinical trials which demonstrate no clear survival advantage for either treatment approach. \cite{bill2005radical, hamdy201610}Hadley found IV to be a useful technique as compared to for example propensity scores in adjustment for confounding in such data. Proportional hazards model was used in their analysis which, as pointed out by \citet{li2015instrumental} as well as explained above, due to the noncollapsibility is not suitable for the two-stage approaches. 
 We would like to instead consider the 
additive hazards model for the reasons given earlier. Since the treatment choices are binary, we would like to use the 2SRI estimator, both for overall survival, and for cancer specific mortality. For these purposes we need to develop the inference procedure under the model for both general survival data and under competing risks. 

The rest of the paper is organized as follows. In section \ref{sec:survival} we describe the assumptions needed for the 2SRI approach under the additive hazards model for general survival data with right censoring, and we study the asymptotic behavior including consistency and asymptotic normality of the 2SRI estimator. Following that, we extend the results to competing risks data in Section \ref{sec:competing} under subdistribution hazard modeling. For both settings we provide a closed-form variance estimate of the 2SRI estimator. Section \ref{sec:simulation} contains finite sample simulation results, and Section \ref{sec:realdata} the analysis of the SEER-Medicare data. Section \ref{sec:discussion} contains some further discussion. All technical details are provided in the Appendix.

\section{Additive hazards model for survival data}\label{sec:survival}



In the presence of possible right censoring, let $T$ and $C$ be the failure time and the censoring time random variables, respectively. We can only observe $T^* = \min(T, C)$ and $\delta = \mathbbm{1}_{\{T \le C\}}$. Similar to the setting in Figure \ref{fig:causal}, denote $\xe$ as the exposure variable, whose causal effect is of primary interest, $\xi$ as the IV, and $\xo$ as the (vector of) observed confounders. 
Our observed data for each individual is $\{T^*_i, \delta_i, \xei, \xoi, \xii\}$ ($i=1, ..., n$), which we assume are independent and identically distributed. In this section, we will assume that $T$ and $C$ are independent conditional on $\xe$, $\xi$ and $\xo$. 
Under the additive hazards model \cite[]{aalen1980model, aalen1989linear, lin1994semiparametric}, the hazard function of $T$ 
given $\xe, \xi, \xo$ and the unobserved confounders 
 is assumed to be in the form
\eqn \label{eq:survivalhazards}
\lambda(t|\xe,\xi,\xo,\xu) = \lambda_0(t) + \be\xe + \bo^\top \xo + \xu,
\ee
where $\xu$ is a function of the unobserved confounders. We assume that $\xu$ is independent of $\xo$ and $\xi$. 
Denote
\eqn \label{eq:first}
\Delta = \xe - \E(\xe|\xi, \xo).
\ee
Following \citet{tchetgen2015instrumental} we put a key assumption on $\xu$:
\eqn\label{eq:errormodel}
\xu = \rho_0 \Delta + \epsilon.
\ee
where 
$\epsilon$ is an error term independent of $\xe$, $\xi$ and $\xo$. 
Proposition \ref{prp:survivalcausal} in the Appendix shows that integrating out $\xu$ we have
\eqn\label{eq:truesurvivalmodel}
\lambda(t| \xe,\xi,\xo) = \bar{\lambda}_0(t) + {\beta}_e\xe + \boldsymbol{{\beta}_o}^\top \xo +\rho_0 \Delta.
\ee
Note that the same coefficient $\beta_e$ (and $\boldsymbol{{\beta}_o} $) from \eqref{eq:survivalhazards} is remained in \eqref{eq:truesurvivalmodel}. 

The error term $\Delta$ in \eqref{eq:truesurvivalmodel} is not readily available from the data. Nonetheless we can `estimate' $\Delta$ and use this estimate as a substitute. For this we need to impose an assumption on the form of $\E(\xe|\xi, \xo)$, for example, 
\eqn\label{eq:firststep}
g(\E(\xe|\xi, \xo)) = \alpha_c + \ai\xi + \ao^\top \xo,
\ee
where $g(\cdot)$ is a link function.

The two stage residual inclusion (2SRI) estimator is then defined as follows: in the first stage, we fit model \eqref{eq:firststep}
and obtain
\eqn \label{eq:5}
\hat{\Delta} = \xe - \hat{X}_e=\xe - \hat{\E}(\xe|\xi, \xo).
\ee
Then in the second stage, we fit \eqref{eq:truesurvivalmodel} with $\Delta$ replaced by $\hat{\Delta}$.

Denote $Z_i = [\xei, \boldsymbol{X}^\top_{oi}, \hat{\Delta}_i]^\top$ the regressors in \eqref{eq:truesurvivalmodel} with $\Delta$ replaced by $\hat{\Delta}$, $\tX_i = [1, \xii, \xoi^\top]^\top$ the regressors in \eqref{eq:firststep}.
Let $\Nit = \mathbbm{1}_{\{T^*_i \le t, \delta_i = 1\}}$ be the counting process, and $\Yit = \mathbbm{1}_{\{T^*_i \ge t\}}$ the at-risk process. Define the filtration $\cF_t = \sigma\{N_i(u), Y_i(u), \xii, \xoi, \xei, u \leq t, i=1, .., n\}$. By the usual counting process theory, $M_i(t) = N_i(t) - \int_0^t Y_i(u) \lambda_i(u) du$ is a mean zero martingale with respect to the filtration $\cF_t$. Under the additive hazards model the estimating equation for $\bbeta = ({\beta}_e, \boldsymbol{{\beta}^\top_o}, \rho_0)^\top$
in \eqref{eq:truesurvivalmodel} is
\eqn
U(\bbeta) = \frac 1n \sum_{i = 1}^n \int_0^1 (Z_i - \bar{Z}(t))( d\Nit - \Yit \bbeta^\top Z_i dt),
\ee
where $\displaystyle{\bar{Z}(t) = {\sum_{l = 1}^n Z_l Y_l(t)} / {\sum_{l = 1}^n Y_l(t)}}$. 
This gives our estimator 
\eqn
\hat{\bbeta} =\Big\{\sum_{i = 1}^n \int_0^1 Y_i(t)(Z_i - \bar{Z}(t))^{\otimes 2}dt\Big\}^{-1}\Big\{\sum_{i = 1}^n \int_0^1 (Z_i - \bar{Z}(t))dN_i(t)\Big\}.
\ee
In the following we show that $\hat\bbeta$ is consistent for the true $\bbeta$ and therefore $\hat\be$ is consistent for the causal parameter $\be$. The estimator is also asymptotically normal and we provide a closed form expression for its asymptotic variance. 

In addition to $\bbeta$, the cumulative baseline hazard function $\Lambda_0(t) = \int_0^t \bar \lambda_0(s)ds$ can be estimated by 
\eqn\label{regbaseline}
\hat{\Lambda}_0(t) = \sum_{i = 1}^n\int_0^t \frac{1}{\sum_{j = 1}^n Y_j(u)}dN_i(u) - \hat{\bbeta}^\top\int_0^t \bar{Z}(u)du.
\ee
Using this we can also estimate the conditional survival function given the observed variables $\boldsymbol{x}= (x_e, \boldsymbol{x}^\top_o, x_I)^\top$, the value of the variables of a future patient whose survival we are interested in predicting:
\eqn\label{regsurvival}
\hat{S}(t|{\boldsymbol{x}}) = \exp(-\hat{\Lambda}_0(t) - \hat{\bbeta}^\top {\boldsymbol{z}}t),
\ee
where $\boldsymbol{z} = (x_e, \boldsymbol{x}^\top_{o}, x_e - \hat{\E}(\xe|x_I, \boldsymbol{x}_{o}))^\top$. 
Note that under the additive hazards model $\hat{\Lambda}_{0}(t)$ can be negative, 
or the estimated survival function $\hat{S}(t|\boldsymbol{x})$ not decreasing. 
Therefore we follow the approach of \citet{lin1994semiparametric} and use a modified $\hat{\Lambda}^*_{0}(t) = \max_{0 \le s \le t}\hat{\Lambda}_{0}(s)$, and $\hat{S}^*(t|\boldsymbol{x}) = \min_{0 \le s \le t}\hat{S}(s|\boldsymbol{x})$. Under regularity condition, the modified version is
asymptotically equivalent to the original version. Now we state our main results below.



\begin{thm}\label{thm:regconsisor}
Under \eqref{eq:survivalhazards}, \eqref{eq:errormodel}, \eqref{eq:firststep} and regularity conditions \ref{ass:bound}, \ref{ass:asyregsurv} and \ref{ass:postivesurv} given in the Appendix, the two stage residual inclusion estimator $\hat{\bbeta}$ is consistent for the true value of $\bbeta$ in \eqref{eq:truesurvivalmodel}, denoted by $\bbeta_{T} $, i.e.~
$\hat{\bbeta} \rightarrow \bbeta_{T}$ in probability as $n\rightarrow \infty$. 
\end{thm}

\begin{thm}\label{thm:regnormal}
Under \eqref{eq:survivalhazards}, \eqref{eq:errormodel}, \eqref{eq:firststep} and Conditions \ref{ass:bound},\ref{ass:asyregsurv},\ref{ass:postivesurv}, $\sqrt{n}(\hat{\bbeta} - \bbeta_{T})$ is asymptotically normally distributed with asymptotic covariance matrix that can be consistently estimated by $\hat{\Omega}^{-1}(\hat{\Sigma}_1 + \hat{\Sigma}_2)\hat{\Omega}^{-1}$, where
\eqn
\hat{\Omega} &=& \frac{1}{n}\sum_{i = 1}^n \int_0^1 Y_i(t)(Z_i - \bar{Z}(t))^{\otimes 2}dt,\\
\hat{\Sigma}_1 &=& \frac{1}{n}\sum_{i = 1}^n \int_0^1 (Z_i - \bar{Z}(t))^{\otimes 2} dN_i(t),\\
\hat{\Sigma}_2 &=& \hat{\Psi} \hat{\Theta} \hat{\Psi}^\top,\label{survsigma2hat}\\
\hat{\Psi} &=& \frac{\hat{\rho}_0}{n}\sum_{i = 1}^n \Big\{\int_0^1 \Yit(Z_i - \bar{Z}(t))dt\Big\}\tX_i^\top(g^{-1})'(\tX_i^\top\hat{\alpha}),
\ee
$\alpha = (\alpha_c, \ai, \ao^\top)^\top$, 
$\hat{\Theta}$ is the estimated covariance matrix of $\sqrt{n}(\hat{\alpha} - \alpha_T)$ from the first stage, and $(g^{-1})'$ is the derivative of the inverse function of $g$. 
\end{thm}

\begin{rem}
Note that $\hat{\Theta}$ can typically be obtained when using software for fitting linear or generalized linear regression models in the first stage. 
For 
linear regression of $\xe$ on $\xi$ and $\xo$ in \eqref{eq:firststep}, $g(y) = y$, so $(g^{-1})' \equiv 1$.
For logistic regression 
$(g^{-1})'(y) = {\exp(y)} / \{1 + \exp(y)\}^2$.
\end{rem}

\begin{thm}\label{thm:regsurvival}
For a new observation $\boldsymbol{x}$, 
Under \eqref{eq:survivalhazards}, \eqref{eq:errormodel}, \eqref{eq:firststep} and Conditions \ref{ass:bound},\ref{ass:asyregsurv},\ref{ass:postivesurv}, the estimated survival function in \eqref{regsurvival} converges to $S(t|x)$ uniformly and the process $\sqrt{n}\{\hat{S}(\cdot|\boldsymbol{x}) - S(\cdot |\boldsymbol{x})\}$ converges weakly to a zero-mean Gaussian process whose covariance function at $(t, s)$, where $0 \le s \le t$, can be consistently estimated by
\eqn
&&\hat{S}(t|\boldsymbol{x})\hat{S}(s|\boldsymbol{x})\Big\{n\sum_{i = 1}^n\int_0^s \frac{1}{(\sum_{j = 1}^n Y_j(u))^2}dN_i(u) + \hat{G}^\top(t) \hat{\Omega}^{-1}(\hat{\Sigma}_1 + \hat{\Sigma}_2)\hat{\Omega}^{-1} \hat{G}(s) \nonumber\\
&&+ \hat{E}^\top(t)\hat{\Theta}\hat{E}(s) + \hat{G}^\top(t) \hat{\Omega}^{-1}\hat{D}(s) + \hat{G}^\top(s) \hat{\Omega}^{-1}\hat{D}(t)\Big\},
\ee
where
\eqn
\hat{D}(t) &=& \sum_{i = 1}^n\int_0^t \frac{ Z_i - \bar{Z}(u)}{\sum_{j = 1}^n Y_j(u)}dN_i(u),\\
\hat{E}(t) &=& \hat{\rho}_0\sum_{i = 1}^n\tX_i(g^{-1})'(\tX_i^\top\hat{\alpha})\int_0^t\frac{Y_i(u) }{\sum_{j = 1}^n Y_j(u)}du,\\
\hat{G}(t) &=& \int_0^t(\boldsymbol{z} - \bar{Z}(u))du.
\ee
\end{thm}
\begin{rem}
In forming the confidence interval (CI) for ${S}(t|\boldsymbol{x}) $, we can take the log-log transformation of $\hat{S}(t|\boldsymbol{x})$ and use Delta method to obtain the confidence interval of $\log \Lambda(t|\boldsymbol{x})$.
 This way the transformed-back confidence interval of $\hat{S}(t|\boldsymbol{x})$ is guaranteed to be within the range of $[0, 1]$. 
\end{rem}
The asymptotic results for the cumulative baseline hazard estimator is in the appendix.

\section{Competing risks}\label{sec:competing}

We now consider competing risks data. As before let $T$ and $C$ be the failure time and the censoring time, respectively. In addition let $J \in \{1,...,K\}$ to be the indicator for cause of failure, and $J = 1$ will be our cause of interest. 
Denote $\boldsymbol{X} = (\xe,\xi,\xo,\xu)$, $F_1(t|\boldsymbol{X}) = P(T\leq t, J = 1 |\boldsymbol{X} )$, and let
$ 
\lambda_1(t|\boldsymbol{X}) = -d\log\{1 - F_1(t|\boldsymbol{X})\}/dt
$ 
be the subdistribution hazard. 
In principle we may assume that $C$ and $T$ are independent conditional on all the observed covariates, but for the estimation approach below we will make use of the marginal Kaplan-Meier estimate of the distribution of $C$. Therefore we will make the stronger assumption that $C$ and $T$ are independent; we will discuss the relaxation of this assumption later. 

Similar to Section \ref{sec:survival} we assume that 
\eqn \label{eq:comphazards}
\lambda_1(t|\xe,\xi,\xo,\xu) = \lambda_{10}(t) + \be\xe + \bo\xo + \xu.
\ee
This is the additive subdistribution hazards model. 
Keeping the same notation as in \eqref{eq:first} and assumption \eqref{eq:errormodel}, we have according to Proposition \ref{prp:compcausal}, 
\eqn\label{eq:truecompmodel}
\lambda_1(t|\xe,\xi,\xo) = \bar{\lambda}_{10}(t) + {\beta}_e\xe + \boldsymbol{\bar{\beta}_o}^\top \xo +\rho_0\Delta.
\ee
Note that although the derivation of Proposition \ref{prp:compcausal} is similar to that of Proposition \ref{prp:survivalcausal}, this is a new result to our best knowledge and the 2SRI approach has not been previously considered under competing risks in the literature. 

In the 2SRI approach for competing risks data here, the first stage is the same as that in Section \ref{sec:survival}, and we replace $\Delta$ by $\hat\Delta$ to fit \eqref{eq:truecompmodel}.
Our observed data are $\{T^*_i, \delta_i, \delta_iJ_i, \xei, \xoi, \xii\}_{1 \le i \le n}$. 
The following are common quantities used in the regression modeling and inference of the subdistribution hazard function.
With a slight abuse of notation in this section, define the event time process as
$\Nit = \mathbbm{1}_{\{T_i \leq t, J_i = 1\}}$, and the at-risk process as
$\Yit = 1 - N_i(t-)$ (note that these have different meanings from Section \ref{sec:survival}). 
Let $\rit = \mathbbm{1}_{\{C_i \geq T_i \wedge t\}}$ denote an individual not censored or yet, so that both $\rit \Nit$ and $\rit \Yit$ are computable from the observed data at any time $t$. Define $w_i(t) = \rit\Gt/\Gtit$, 
 and $\hwit = \rit\hGt/\hGtit$, where $\Gt = \P(C \geq t)$ and $\hGt$ is the Kaplan-Meier estimate for $\Gt$ using $\{T_i^*, 1 - \delta_i\}_{1 \le i \le n}$. 
 
 We will use $\Mi1t$ to denote the martingale for the $i$-th object with respect to the complete-data filtration, that is, $\cF^1(t) = \sigma\{N_i(u), Y_i(u),\xei, \xii, \xoi, u \le t, \forall ~ 1 \le i \le n\}$. We will also use $M_i^c(t)$ to denote the martingale for the censoring related process of the $i$-th subject, $M_i^c(t) = N_i^c(t) - \int_0^t \mathbbm{1}_{\{T_i^* \ge u\}}d\Lambda^c(u)$, where $N_i^c(t) = \mathbbm{1}_{\{T_i^* \leq t, \delta_i = 0\}}$ is the censoring counting process, $\Lambda^c(t)$ is the cumulative hazard function of the censoring distribution. The censoring filtration is $\cF^c(t) = \sigma\{\mathbbm{1}_{\{T_i^* \ge u\}}, \mathbbm{1}_{\{T_i^* \leq u, \delta_i = 0\}}, \xei, \xii, \xoi, u \le t, \forall ~ 1 \le i \le n\}$.

The estimating function for $ \bbeta = (\beta_e, \beta^\top_o, \rho_0)^\top$ can be written as \cite{zheng2017instrumental, li2017additive}
\eqn
U(\bbeta) = \frac{1}{n} \sum_{i = 1}^n \int_0^1 (Z_i - \bar{Z}(t))\hwit(d\Nit - \Yit \bbeta^\top Z_i dt),
\ee
where $Z_i$ is the same as defined in Section \ref{sec:survival}, $\displaystyle{\bar{Z}(t) = \frac{\sum_{l = 1}^n Z_l \hat{w}_l(t) Y_l(t)}{\sum_{l = 1}^n \hat{w}_l(t) Y_l(t)}}$.
Therefore,
\eqn
\hat{\bbeta} = \Big\{\sum_{i = 1}^n \int_0^1 \hwit Y_i(t)(Z_i - \bar{Z}(t))^{\otimes 2}dt\Big\}^{-1}\Big\{\sum_{i = 1}^n \int_0^1 (Z_i - \bar{Z}(t))\hwit dN_i(t)\Big\}.
\ee
 The baseline cumulative hazard function $ {\Lambda}_{10} = \int_0^\cdot {\lambda}_{10}$ is then estimated by
\eqn \label{compbaseline}
\hat{\Lambda}_{10}(t) = \sum_{i = 1}^n\int_0^t \frac{\hat{w}_i(u) }{\sum_{j = 1}^n \hat{w}_j(u) Y_j(u)}dN_i(u) - \hat{\bbeta}^\top\int_0^t \bar{Z}(u)du.
\ee
Therefore the estimated cumulative incidence function (CIF) is 
\eqn \label{compsurvival}
\hat{F}_1(t|\boldsymbol{x}) = 1 - \hat{S}_1(t|\boldsymbol{x}) = 1 - \exp(-\hat{\Lambda}_{10}(t) - \hat{\bbeta}^\top\boldsymbol{z}t), 
\ee
where $\boldsymbol{z} = (x_e, \boldsymbol{x}^\top_{o}, x_e - \hat{\E}(x_e|x_I, \boldsymbol{x}_{o}))^\top$, and $(x_e, \boldsymbol{x}_o, x_I)$ is the value of the variables of a future patient whose CIF we are interested in predicting. We use the same modified version $\hat{\Lambda}^*_{10}(t) = \max_{0 \le s \le t}\hat{\Lambda}_{10}(s)$ and $\hat{F}^*_1(t|\boldsymbol{x}) = \max_{0 \le s \le t}\hat{F}_1(s|\boldsymbol{x})$ as in Section \ref{sec:survival} to ensure that the estimated hazard is non-negative. Now we state our main results below.


\begin{thm}\label{thm:compcons}
Under \eqref{eq:comphazards}, \eqref{eq:errormodel}, \eqref{eq:firststep} and Conditions \ref{ass:bound},\ref{ass:asyregcomp},\ref{ass:postivecomp}, the two stage residual inclusion estimator $\hat{\bbeta}$ is consistent for the true value of $\bbeta$ in \eqref{eq:truesurvivalmodel}, denoted by $\bbeta_{T} $, i.e.~
$\hat{\bbeta} \rightarrow \bbeta_{T}$ in probability as $n\rightarrow \infty$. 
\end{thm}

\begin{thm}\label{thm:compnormal}
Under \eqref{eq:comphazards}, \eqref{eq:errormodel}, \eqref{eq:firststep} and Conditions \ref{ass:bound},\ref{ass:asyregcomp},\ref{ass:postivecomp}, $\sqrt{n}(\hat{\bbeta} - \bbeta_{T})$ is asymptotically normally distributed with asymptotic covariance matrix that can be consistently estimated by $\hat{\Omega}^{-1}(\hat{\Sigma}_1 + \hat{\Sigma}_2 + \hat{\Sigma}_3)\hat{\Omega}^{-1}$, where
\eqn
\hat{\Omega} &=& \frac{1}{n}\sum_{i = 1}^n \int_0^1 \hwit Y_i(t)(Z_i - \bar{Z}(t))^{\otimes 2}dt,\\
\hat{\Sigma}_1 &=& \frac{1}{n}\sum_{i = 1}^n \int_0^1 (Z_i - \bar{Z}(t))^{\otimes 2}\hwit dN_i(t),\\
\hat{\Sigma}_2 &=& \hat{\Psi} \hat{\Theta} \hat{\Psi}^\top,\\
\hat{\Sigma}_3 &=& \frac{1}{n}\sum_{i = 1}^n \int_0^1 \Big(\frac{\hat{q}(t)}{\hat{\pi}(t)}\Big)^{\otimes 2} dN_i^c(t),\\
\hat{\Psi} &=& \frac{\hat{\rho}_0}{n}\sum_{i = 1}^n \Big\{\int_0^1 \hwit\Yit(Z_i - \bar{Z}(t))dt \Big\}\tX_i^\top (g^{-1})'(\tX_i^\top \hat{\alpha}),\\
\hat{q}(t) &=& -\frac{1}{n}\sum_{i = 1}^n\int_0^1 \mathbbm{1}_{\{ T^*_i < t \le u\}}\hat{w}_i(u)\big(Z_i - \bar{Z}(u)\big) d\hat{M}_i(u),\\
\hat{\pi}(t) &=& \frac{1}{n}\sum_{i = 1}^n \mathbbm{1}_{\{T^*_i \geq t\}},\\
\hat{M}_i(t) &=& N_i(t) - \int_0^t Y_i(u)(d\hat{\Lambda}_{10}(u) + \hat{\bbeta}^\top Z_i du),
\ee
and $\hat{\Theta}$ is the estimated variance-covariance matrix of $\sqrt{n}(\hat{\alpha} - \alpha_T)$ from the first stage. 
\end{thm}

\begin{thm}\label{thm:compsurvival}
For a new observation $\boldsymbol{x}$, under \eqref{eq:comphazards}, \eqref{eq:errormodel}, \eqref{eq:firststep} and Conditions \ref{ass:bound},\ref{ass:asyregcomp},\ref{ass:postivecomp}, the estimated CIF in \eqref{compsurvival} converges to $F_1(t|x)$ uniformly and the process $\sqrt{n}\{ \hat{F}_1(\cdot|\boldsymbol{x}) - F_1(\cdot|\boldsymbol{x})\}$ converges weakly to a zero-mean Gaussian process whose covariance function at $(t, s)$, where $0 \le s \le t$, can be consistently estimated by 
\eqn
&&(1 - \hat{F}_1(t|\boldsymbol{x}))(1 - \hat{F}_1(s|\boldsymbol{x}))\Big\{ \int_0^s \frac{n\sum_{i = 1}^n \hat{w}^2_i(u) dN_i(u)}{(\sum_{j = 1}^n \hat{w}_j(u) Y_j(u))^2} + \hat{G}^\top(t) \hat{\Omega}^{-1}(\hat{\Sigma}_1 + \hat{\Sigma}_2 + \hat{\Sigma}_3)\hat{\Omega}^{-1} \hat{G}(s) \nonumber\\
&&+ n\sum_{i = 1}^n\int_0^1 \frac{\hat{q}_t(u)\hat{q}_s(u)}{\hat{\pi}^2(u)}dN_i^c(u) + \hat{E}^\top(t)\hat{\Theta}\hat{E}(s) + \hat{G}^\top(t) \hat{\Omega}^{-1}\hat{D}(s) + \hat{G}^\top(s) \hat{\Omega}^{-1}\hat{D}(t)\Big\},
\ee
where
\eqn
\hat{D}(t) &=& \sum_{i = 1}^n\int_0^t \frac{(Z_i - \bar{Z}(u))\hat{w}_i(u)}{\sum_{j = 1}^n \hat{w}_j(u) Y_j(u)}dN_i(u),\\
\hat{E}(t) &=& \hat{\rho}_0\sum_{i = 1}^n\tX_i(g^{-1})'(\tX_i^\top \hat{\alpha})\int_0^t\frac{\hat{w}_i(u)Y_i(u) }{\sum_{j = 1}^n \hat{w}_j(u)Y_j(u)}du,\\
\hat{G}(t) &=& \int_0^t(\boldsymbol{z} - \bar{Z}(u))du,\\
\hat{q}_t(u) &=& \frac{1}{n}\sum_{i = 1}^n\int_0^t \frac{\mathbbm{1}_{\{T_i^* < u \le v\}}\hat{w}_i(v)}{\sum_{j = 1}^n \hat{w}_j(v)Y_j(v)}d\hat{M}_i(v).
\ee
\end{thm}
\begin{rem}
As in Section 2 we can take the log-log transformation of $1 - \hat{F}_1(t|\boldsymbol{x})$, and use Delta method to obtain the confidence interval of $\log \Lambda_1(t|x)$. This way the transformed-back confidence interval of $1 - \hat{F}_1(t|\boldsymbol{x})$ is guaranteed to be within the range of $[0, 1]$.
\end{rem}
The asymptotic results for the cumulative baseline hazard estimator is in the appendix.

\section{Simulation}\label{sec:simulation}

We are in the process of completing an R package for our estimators. The following numerical results were obtained using the program which is the core of the package. 

To study the performance of our estimators under both survival and competing risks settings, we carried out simulation studies with  sample size $100, 200, 400, 800, 1200$, and repeated 1000 times for each sample size. We provided in the tables the bias of the estimator, the empirical variance of the estimator from the 1000 repeats, the mean of the variance estimate, and the coverage rate of the nominal 95\% confidence intervals.

\subsection{Regular survival model}

For this part without competing risks, we considered the following three scenarios.

Scenario I: We sampled $\xi, X_o$ from independent standard normal distributions, set $\alpha = [1, 1, 0.5]^\top$, and generated $\xe = [1, \xi, X_o]\alpha^\top + \Delta$, where $\Delta \sim N(0, 0.2)$ was independent of $\xi, X_o$. We simulated $X_u$ according to \eqref{eq:errormodel}
with independent $\epsilon_i \sim N(0, 0.1)$, and $\rho_0 = 1$. We set $\beta = [1, 0.5, 1.5]^\top$ and the baseline hazard $\lambda_0 (t) \equiv 10.5$ in \eqref{eq:survivalhazards} to generate the survival time $T$. There was no censoring in this case.

Scenario II: Similar to I above, but with $\alpha = [0.25, 0.3, 0.2]^\top$, $\beta = [0.5, 0.2, 0.3]^\top$, and $\lambda_0(t) \equiv 5t + 5$. 
The censoring distribution followed exponential with rate of $2$ so that the censoring rate is around $40\%$.

Scenario III: We generated binary exposure as follows. 
First we sampled $\xi$ from Bernoulli distribution with $P(\xi=1) = 0.5 $, and independent $X_o$ from standard normal distribution. We generated $X_e$ from 
$$
X_{ei}|X_{Ii}, X_{oi} \sim \text{Bern}\Big(\frac{1}{1 + \exp(-(\alpha_0 + \alpha_i X_{Ii} + \alpha_o X_{oi}))}\Big).
$$
with $\alpha = [1, 0.5, 1]^\top$. We then simulated 
$$
X_{ui} = \rho_0\{ X_{ei} - E(X_{ei}|X_{Ii}, X_{oi}) \} + \epsilon_i,
$$
with $\rho_0 = 1$ and independent $\epsilon_i \sim N(0, 0.1)$. Finally we generate the survival time by setting $\beta = [1, 0.5, 1.5]^\top$ and the baseline hazard $\lambda_0 (t) \equiv 10.5$. The censoring distribution was exponential with rate of $5$ so that the censoring rate is around $30\%$.

\subsection{Competing risks model}
For the competing risks model, let 
\eqn
\P(J = 1|\boldsymbol{X}) = 1 - \exp(-\int_0^{t_0}\lambda_{10}(u)du - (\beta_{e1}\xe + \beta_{o1}X_o + \beta_{u1}X_u)t_0),
\ee
and $\P(J = 2|\boldsymbol{X}) = 1- \P(J = 1|\boldsymbol{X}) $. Here $t_0$ is a maximum follow-up time. 
We then simulated the event time data from the following subdistributions:
\eqn
F(t|J = 1, \boldsymbol{X}) = \frac{1 - \exp(-\int_0^{\min(t, t_0)}\lambda_{10}(u)du- (\beta_{e1}\xe + \beta_{o1}X_o + \beta_{u1}X_u)\min(t, t_0))}{1 - \exp(-\int_0^{t_0}\lambda_{10}(u)du - (\beta_{e1}\xe + \beta_{o1}X_o + \beta_{u1}X_u)t_0)},
\ee
\eqn
F(t|J = 2, \boldsymbol{X}) = \frac{1 - \exp(-(\lambda_{20} + \beta_{e2}\xe + \beta_{o2}X_o + \beta_{u2}X_u)\min(t, t_0))}{1 - \exp(-(\lambda_{20} + \beta_{e2}\xe + \beta_{o2}X_o + \beta_{u2}X_u)t_0)}.
\ee
Note that
\eqn
F(t, J = 1|\boldsymbol{X}) = 1 - \exp(-\int_0^{\min(t, t_0)}\lambda_{10}(u)du- (\beta_{e1}\xe + \beta_{o1}X_o + \beta_{u1}X_u)\min(t, t_0)).
\ee
Note that the maximum follow-up time $t_0$ is necessary because under the additive hazards model 
$
\Lambda_1(t|Z) = \Lambda_{10}(t) + \bbeta^\top Z t
$ 
 goes to infinity as $t\rightarrow \infty$, implying that $\lim_{t\rightarrow \infty} F_1(t|Z)=1$, which is no longer a subdistribution function.


Scenario I: Similar to Scenario I under the regular survival model but with $\alpha = [1.5, 1, 0.7]^\top$, $\beta_1 = [1, 0.5, 0.75]^\top$ and $\lambda_{10}(t) \equiv 11$ for cause 1, and $\beta_2 = [1.2, 1, 1.3]^\top$ and $\lambda_{20} = 15$ for cause 2. 
We set $t_0 = 0.095$. There was no censoring, and the cause 1 event rate was around $60\%$.

Scenario II: similar to I above, but with $\alpha = [1, 1, 0.5]^\top$, $\beta_1 = [1, 0.5, 0.75]^\top$ and $\lambda_{10}(t) = 5t + 10$ for cause 1, and 
the same as in Scenario I for cause 2. 
We set $t_0 = 0.06$. The censoring distribution was exponential with rate of $1$. The cause 1 event rate was around $28\%$ and censoring rate around $38\%$.

Scenario III: Similar to Scenario III under the regular survival model, but with $\alpha = [-1, 2, 1]^\top$, $\beta_1 = [1, 0.5, 0.75]^\top$ and $\lambda_{10}(t) \equiv 10$ for cause 1, and the same as in Scenario I, II for cause $2$.
 We set $t_0 = 0.06$. The censoring distribution was exponential with rate of $25$. The cause 1 event rate was around $26\%$ and censoring rate aroung $44\%$.

\subsection{Simulation results}

The results are summarized in Table \ref{table:surv} and \ref{table:comp}, respectively. From the tables we can see that as the sample size increases the bias goes down, and the variance estimate also becomes closer to the empirical variance. The coverage rate is quite close to nominal $95\%$ level in all cases. The variance of the estimator becomes larger in Scenario III both with or without competing risks, probably because we have a binary treatment, which leads to a error term with large variance, compared to the error terms in other scenarios. Nonetheless the variance estimate still works well. In general our estimator behaves well under all the finite-sample settings considered.

\section{SEER-Medicare data analysis}\label{sec:realdata}

For this analysis we consider prostate cancer patients with localized non-metastatic disease identified from the linked SEER-Medicare database diagnosed between 2000-2011 and followed up through 12/31/2013. 
The variables included were age, race/ethnicity, marital status, tumor stage, tumor grade, Prior Charlson comorbidity score measured during the year prior to diagnosis, year of diagnosis, and hospital referral regions. The hospital referral region was an important variable for us to construct the instrumental variable. Hospital referral regions represents a set of contiguous zip codes around a major hospital. 
Following \citet{hadley2010comparative} we restricted the analysis to early stage (T1 and T2) patients, aged 66 to 74 years, as well as eliminated patients in geographic areas with fewer than 50 patients over the entire observation period. This led to an overall sample size of $n=29806$. Among them $493$ ($1.65\%$) patients died due to cancer, 
$2066$ ($6.93\%$) died due to other causes, and the remaining $27247$ ($91.4\%$) were alive at the end of the follow-up.
There were four types of treatments: surgery, radiation, chemotherapy and hormonal therapy; $10977$ people received surgery, $21357$ radiation, $9577$ chemotherapy, and $9527$ hormonal therapy. Note that some patients received more than one treatment. Following \citet{hadley2010comparative} we will label patients who received surgery as ``radical prostatectomy" and the remaining ``conservative management". We will then compare the effects of these two treatments on the time to death for all causes and due to cancer, respectively.
 A summary of the patient characteristics is presented in Table \ref{table:analysis}. It can be seen that patients who received surgery tended to be younger, married, non-black, have T2 stage, well differentiated tumor grade, comorbidity score 0, and diagnosed no later than 2005. 
 \citet{hadley2010comparative} showed that the treatment pattern varied by hospital referral regions, beyond what was captured by the patient characteristics in Table \ref{table:analysis}.

For comparison purposes we first fitted the additive hazards model including the variables in Table \ref{table:analysis} and all their pairwise interactions, but without using any IV. It turned out that the treatment had a significant effect with $p$-value of $0.001$ for all causes of death. On the other hand the treatment effect was not significantly different from zero for our data with $p$-value of $0.13$ for cancer specific survival. We note that in \citet{hadley2010comparative} the treatment had a significant effect on both overall survival and cancer specific survival; our data was a later export than those used by \citet{hadley2010comparative} (diagnosed between 1995 and 2003) from the linked database. 
 In addition, \citet{hadley2010comparative} used the Cox multiplicative hazards model as opposed to our additive hazards model. 

We now consider the 2SRI approach. We used the same instrumental variable as in \citet{hadley2010comparative}. Specifically, we constructed the IV as follows. We first applied logistic regression to obtain the predicted probability for conservative management given covariates including age, race/ethnicity, marital status, tumor stage, tumor grade description, Prior Charlson comorbidity score, year of diagnosis and all the two-way interactions. Then for each hospital referral region and each year, we calculated the difference between the proportion of patients receiving conservative management and the average predicted probability of conservative management. 
Clearly, a larger difference indicated that the corresponding hospital referral region favored the conservative management more than those with a smaller difference. Therefore this difference was likely correlated with the treatment a patient received and, on the other hand, this difference was unlikely to directly influence the survival of an individual patient beyond the treatment assignment.
For use as an IV we lagged this difference for one year for the patients coming from the same hospital referral region. Therefore, the data that we used to analysis survival were patients diagnosed from 2001 to 2011. 

We then performed the first step of the IV analysis, using logistic regression of treatment on the IV obtained above, together with the other observed confounders including age, race or ethnicity, marital status, tumor stage, grade description and Prior Charlson comorbidity score, year of diagnosis and all two-way interactions. We found that the instrument had an extremely significant effect on the treatment with $p$-values less than $2 \times 10^{-16}$, indicating likely a strong IV. 
We then subtracted the predicted probability of treatment from the observed treatment to obtained the residuals. In the second step, we included this residual term together with the treatment and all the confounders and their pairwise interactions to fit the survival models. The results for overall survival are shown in Table \ref{table:ivallcauseresult}, and for cancer specific survival in Table \ref{table:cancercauseresult}.


From the tables we see that the causal effect of treatment remained?? 
significant for overall survival, although the $p$-value increased from 0.001 to 0.042. The $p$-value for the causal effect of treatment on cancer specific survival also increased from 0.13 to 0.41. The differences between the IV analysis results and the initial analysis results earlier indicate that there were likely unobserved confounders for the treatment effect on both overall and cancer specific survival, beyond those captured in Table \ref{table:analysis} (and their interactions).


Finally,  Figure \ref{fig:prediction} illustrates the predicted overall survival as well as cancer specific cumulative incidence function for a patient who received radical prostatectomy, was diagnosed in 2001, aged 71, white, with `other' marital status, tumor stage T2, moderately differentiated tumor, Charlson comorbidity score 2, and instrument value 0.0.3429. In reality the patient survived for 95 months, and died of other causes.

\section{Discussion}\label{sec:discussion}

In this paper we have developed statistical inference procedures for the 2SRI IV estimator under GLM in the first stage and  additive hazards model in the second stage for survival data that was conceptually described in \citet{tchetgen2015instrumental}. We have also extended the approach to competing risks data under the additive subdistribution hazards model. An R package is being completed that computes these estimators and their closed-form estimated asymptotic variances, as well as prediction under these models given the observed covariates. Our simulation results show the satisfactory performance of the procedures, and the SEER-Medicare analysis shows the usefulness of the approaches. 

The causal effect of interest $\beta_e$ that we have considered in this work is conditional on the unobserved confounder $\xu$, although from \eqref{eq:truesurvivalmodel} and \eqref{eq:truecompmodel} we may also understand it as conditional on the observed variables. This seems reasonable in our comparative effectiveness settings, where a relatively large number of observed confounders are typically considered. 
\citet{cai2011two} considered the setting for compliance in randomized clinical trials with binary outcomes, and pointed out that the 2SPS and 2SRI approaches may not estimate the causal odds ratio among compliers under the principal stratification framework.\citep{angrist1996identification, fran:rubin} Future work may consider similar analysis under the additive hazards model.




Under competing risks and using subdistribution hazards modeling, weights based on the estimated censoring distribution are needed in order to consistently estimate the regression coefficients. The marginal Kaplan-Meier estimate we used requires independent censoring. Alternatively one may estimate the conditional distribution of censoring given covariates using, for example, semiparametric survival models. Recently \citet{nguy:gillen} proposed to estimate this conditional distribution nonparametrically using a survival tree approach. When any of these estimates are used, the assumption on censoring distribution can then be relaxed to be conditionally independent of failure time (and type/cause) given the covariates.


All our technical proofs are compatible with time-dependent covariates. However, the causal inference problem is more complex with time-varying confounders and time-varying treatments, especially if the later confounders are affected by the earlier treatments.\citep{hern:etal}
To our best knowledge time-varying instrument variable method has not been developed in the literature.


\section*{Acknowledgement}

The authors would like to thank Mr.~Vinit Nalawade for help with preparing the SEER-Medicare data set.

\bibliographystyle{myunsrtnat}
\bibliography{ref}
\nocite{*}

\newpage
\section{Appendix}\label{sec:app}
\subsection{Preparations}
The following results are either from  or easy consequences of \citet{fleming2011counting}. 
The predictable variation process $ \langle \, \rangle $ and quadratic variation process $[\, ] $ are also defined in \citet{fleming2011counting}. 
\begin{lem}\label{lem:1}
Assume that for each $t \geq 0$, given $\cF_{t-}$, $\{d N_1(t),...,d N_n(t)\}$ are independent 
 $0, 1$ random variables, set $M_j = N_j - A_j$, where $A_j$ is the compensator for $N_j$. Then for any $i \neq j$ and $t \geq 0$, 
\eqn
\langle M_i, M_j\rangle(t) = 0, \quad \text{a.s.}
\ee
\end{lem}

\begin{lem}\label{lem:2}
Let $M(t)$ be a martingale with respect to the filtration $\cF(t)$, $H(t)$ be a predictable process, then the predictable variation process of the martingale integral $\int_0^t H(s)dM(s)$ is
\eqn
\langle \int_0^t H(s)dM(s)\rangle = \int_0^t H^2(s)d\langle M\rangle(s).
\ee
\end{lem}

\begin{lem}\label{lem:3}
Let $M_1(t)$, $M_2(t)$ be  martingales with respect to the filtration $\cF(t)$, $H_1(t)$, $H_2(t)$ be  predictable processes, then the predictable covariation process of the martingale integral $\int_0^t H_1(s)dM_1(s)$ and $\int_0^t H_2(s)dM_2(s)$ is
\eqn
\langle \int_0^t H_1(s)dM_1(s), \int_0^t H_2(s)dM_2(s)\rangle = \int_0^t H_1(s)H_2(s)d\langle M_1, M_2\rangle(s).
\ee
\end{lem}


\begin{lem}\label{lem:variation}
For independent and identically distributed sequences of martingale $\{M_i\}$ and predictable process $H_i(t)$, $1 \le i \le n$, $\Var\Big(\int_0^t H_1(u)dM_1(u)\Big) = \E\Big(\int_0^t H_1(u)dM_1(u)\Big)^2 = \E\Big(\langle\int_0^t H_1(u)dM_1(u)\rangle\Big) = \E\Big(\int_0^t H_1^2(u)d\langle M_1 \rangle(u)\Big)$ can be estimated by
\eqn
\frac{1}{n}\sum_{i = 1}^n \int_0^t H_i^2(u)d[M_i](u).
\ee
\end{lem}

\begin{lem}[Lenglart's inequality] \label{lem:lenglart}
Let W be a local square integrable martingale. Then for all $\delta, \eta >0$, 
\eqn
\P(\sup_{t\in [0,1]}|W(t)| > \eta) \leq\frac{\delta}{\eta^2} + \P(\langle W, W\rangle(1) > \delta).
\ee
\end{lem}
Basically the Lenglart's inequality tells us that the convergence in probability of the supreme of a martingale can be infered from its endpoint.

\begin{lem} \label{lem:martint}
Let $M_i(t)$ be independent and identically distributed martingales with respect to the filtration $\cF(t)$, $i = 1,..., n$, $H_i(t)$ also be i.i.d. predictable processes, then
\eqn
\frac{1}{\sqrt{n}} \sum_{i = 1}^n \int_0^t H_i(s)dM_i(s) \to_p 0 \quad \text{uniformly in } t \in [0,1],
\ee
if each $\sup_{0\leq t\leq 1} |H_i(t)| = o_p(1)$. 
\end{lem}
\begin{proof}
By Lemma \ref{lem:1} and \ref{lem:3}, 
\eqnn
\langle \frac{1}{\sqrt{n}} \sum_{i = 1}^n \int_0^t H_i(s)dM_i(s)\rangle &=& \frac{1}{n}\sum_{i = 1}^n\sum_{j = 1}^n\int_0^t H_i(s)H_j(s)d\langle M_i, M_j\rangle(s)\\
&=& \frac{1}{n}\sum_{i = 1}^n\int_0^t H^2_i(s)d\langle M_i\rangle(s).
\een
Therefore,
\eqnn
\P\Big(\sup_{t\in [0,1]}\big|\frac{1}{\sqrt{n}} \sum_{i = 1}^n \int_0^t H_i(s)dM_i(s)\big| > \eta\Big) \leq\frac{\delta}{\eta^2} + \P\Big(\frac{1}{n}\sum_{i = 1}^n\int_0^1 H^2_i(s)d\langle M_i\rangle(s) > \delta \Big).
\een
Since $\langle M_i \rangle$ is increasing, 
\eqn
\frac{1}{n}\sum_{i = 1}^n\int_0^1 H^2_i(s)d\langle M_i\rangle(s) \leq \frac{1}{n}\sum_{i = 1}^n [\sup_{0 \leq s \leq t} H^2_i(s)] \langle M_i\rangle(t),
\ee
if $\sup_{0\leq t\leq 1} |H_i(t)| = o_p(1)$, so is $\displaystyle{\frac{1}{n}\sum_{i = 1}^n [\sup_{0 \leq s \leq t} H^2_i(s)] \langle M_i\rangle(t)}$ and $\displaystyle{\frac{1}{\sqrt{n}} \sum_{i = 1}^n \int_0^t H_i(s)dM_i(s)}$, the latter by Lemma \ref{lem:lenglart}.
\end{proof}

The following proposition can be easily proved by mimicking the proof of Glivenko-Cantelli theorem.
\begin{prp}\label{prp:sup1}
Assume that $(X_1, Y_1),\dots, (X_n, Y_n)$ are independent and identically distributed pairs of random variables with distribution function $F(x, y)$. Also, $X_1$ has finite first moment. Define
\eqn
G_n(t) = \frac{1}{n}\sum_{i = 1}^n X_i \mathbbm{1}_{\{[Y_i, \infty)\}}(t),
\ee
and 
\eqn
G(t) = \E[X_1 \mathbbm{1}_{\{[Y_1, \infty)\}}],
\ee
then we have a similar result to Glivenko-Cantelli theorem, that is,
\eqn
||G_n - G||_\infty = \sup_{t \in \bbR}|G_n(t) - G(t)| \longrightarrow 0.
\ee
\end{prp}

\subsection{Additive hazards model for survival data}
\begin{prp}\label{prp:survivalcausal}
Assuming \eqref{eq:survivalhazards} and \eqref{eq:errormodel}, we have
\eqn
\bar{\lambda}(t|\xe,\xi,\xo) = \bar{\lambda}_0(t) + \be\xe + \bar{\bo}^\top\xo + \rho_0\Delta,
\ee
where $\bar{\lambda}_0(t) = \lambda_0(t) - \frac{\partial}{\partial t}\log\Big[\E\big\{\exp(-\epsilon t)\}	\Big]$.
\end{prp}
\begin{proof}
By the assumption \eqref{eq:survivalhazards}, we have
\eqnn
S(t|\xe, \xi, \xo, \xu) = \exp\Big\{-\int_0^t \big[\lambda_0(s) + \be\xe + \bo^\top\xo + \xu\big]ds\Big\}.
\een
Therefore, integrating out $\xu$ we have,
\eqnn
&&S(t|\xe,\xi,\xo)\\
&=&\E[S(t|\xe, \xi, \xo, \xu)|\xe, \xi, \xo]\\
&=&\E\big[\exp\big\{-\int_0^t (\lambda_0(s) + \be\xe + \bo^\top\xo + \E(\xu|\xe, \xi, \xo)) \big\}ds |\xe, \xi, \xo\big]\\
&=&\E\big[\exp\big\{-\int_0^t (\lambda_0(s) + \be\xe + \bo^\top\xo + \rho_0 \Delta + \epsilon) ds \big\}|\xe, \xi, \xo\big]\\
&=&\exp\big\{-\int_0^t (\lambda_0(s) + \be\xe + \bo^\top\xo + \rho_0 \Delta\big\} \times \E\big[\exp\{ - \epsilon t\}|\xe, \xi, \xo \big]\\
&=&\exp\big\{-\int_0^t (\lambda_0(s) + \be\xe + \bo^\top\xo + \rho_0 \Delta ) ds\big\} \times \E\big[\exp\{ - \epsilon t\} \big].
\een
Now the hazard function becomes,
\eqnn
\lambda(t|\xe,\xi,\xo) &=& -\frac{\partial}{\partial t}\log S(t|\xe,\xi,\xo)\\
&=& \lambda_0(t) - \frac{\partial}{\partial t}\log\Big[\E\big\{\exp(-\epsilon t )\big\}\Big] + \be\xe + \bo^\top\xo + \rho_0\Delta,
\een
where $\bar{\lambda}(t|\xe,\xi,\xo) = \bar{\lambda}_0(t) + \be\xe + \bar{\bo}^\top\xo + \rho_0\Delta$ is our new baseline hazard function.
\end{proof}

In the following we give the regularity conditions and addition notation for the results of the Theorems in section \ref{sec:survival}. 

\begin{ass}\label{ass:bound} 
The covariates $\{\xei, \xii, \xoi\}$ are bounded. 
\end{ass}
Note that
\eqnn
\hat{\Delta} - \Delta &=& -\big\{\hat{\E}(\xe|\xi, \xo) - \E(\xe|\xi, \xo)\big\} \\
&=& -\big\{g^{-1}(\tX^\top \hat{\alpha}) - g^{-1}(\tX^\top \alpha_T)\big\}
\een
\begin{ass}\label{ass:firststep}
We assume that $\hat{\alpha}$ is consistent for $\alpha_T$, which implies that $|\hat{\Delta} - \Delta| \to_p 0$ with Condition \ref{ass:bound}. We also assume that
\eqnn
\sqrt{n}(\hat{\alpha} - \alpha_T) = \mathcal{I}^{-1}\cdot(\frac{1}{\sqrt{n}}\sum_{i = 1}^n T_i) + o_p(1);
\een
in other words,  $\sqrt{n}(\hat{\alpha} - \alpha_T)$ can be written into a sum of i.i.d. terms with finite variance plus one $o_p(1)$ term. 
\end{ass}
\begin{rem}
This Condition is usually fulfilled by the first step estimator, such as the maximum likelihood estimator under the GLM.
\end{rem}

\begin{ass}\label{ass:baselinebound}
\eqn
\int_0^1\bar{\lambda}_0(t)dt < \infty.
\ee
\end{ass}

Define 
$S^{(j)}(t) = \sum_{i = 1}^n Y_i(t)Z_i^{\otimes j} / n$ for $j = 0$ and 1. Then by Gilvenko-Cantelli theorem there exists a scalar and vector function $\szero(t)$ and $\sone(t)$ defined on $[0,1]$ such that 
\eqn
\sup_{t \in [0,1]}||S^{(j)}(t) - s^{(j)}(t)|| \xrightarrow{\cP} 0.
\ee

\begin{ass}[Asymptotic regularity conditions]\label{ass:asyregsurv}
$\szero(t)$ is bounded away from zero.
\end{ass}

\begin{ass}\label{ass:postivesurv}
There exists a positive definite matrix $\Omega$ such that 
\eqn
\frac 1n\sum_{i = 1}^n \int_0^1 \Yit\Big(Z_{0i} - \frac{\sone(t)}{\szero(t)}\Big)^{\otimes 2}dt \xrightarrow{a.s.} \Omega.
\ee
\end{ass}

Let $Z_{0i} = [X_{ei}, \xoi^\top, \Delta_i]$, $\bar{Z}_0(t) = \sum_{l = 1}^n Z_{0l} Y_l(t)/\sum_{l = 1}^n Y_l(t)$, define
\eqn
\bbA_1 &=& \sum_{i = 1}^n \int_0^1 (Z_i - Z_{0i})d\Mit,\\
\bbA_2 &=& \sum_{i = 1}^n \int_0^1 (\bar{Z}(t) -\bar{Z_0}(t))d\Mit,\\
\bbA_3 &=& \rho_0\sum_{i = 1}^n \big\{\int_0^1 (Z_i - Z_{0i})\Yit dt\big\}(\Delta_i - \hat{\Delta}_i),\\
\bbA_4 &=& \rho_0\sum_{i = 1}^n \big\{\int_0^1 (\bar{Z}(t) -\bar{Z_0}(t))\Yit dt\big\}(\Delta_i - \hat{\Delta}_i).
\ee

\begin{lem}\label{lem:regular}
$\bbA_1$, $\bbA_2$, $\bbA_3$, $\bbA_4$ are bounded in probability.
\end{lem}

\begin{proof}
Notice that 
\eqnn
Z_i - Z_{0i} = [0, 0, \hat{\Delta}_i - \Delta_i]^\top,
\een
where $\hat{\Delta}_i - \Delta_i = g^{-1}(\tX_i^\top\alpha_T) - g^{-1}(\tX_i^\top\hat{\alpha}) = (g^{-1})'(\tX_i^\top\alpha_T) + o_p(1)$. Therefore it suffices to check the $\hat{\Delta}_i - \Delta_i$ components. Then the only nonzero entry of $\bbA_1$ is 
\eqnn
\sum_{i = 1}^n \int_0^1 \{\hat{\Delta}_i - \Delta_i\}d\Mit =\Big\{ -\frac{1}{\sqrt{n}} \sum_{i = 1}^n \int_0^1 \tX_i^\top (g^{-1})'(\tX_i^\top \alpha_T) d\Mit \Big\}\sqrt{n}(\hat{\alpha} - \alpha_T) + o_p(1),
\een
which is bounded in probability by central limit theorem. Next, we check the nonzero entry of $\bbA_2$,
\eqnn
&&\sum_{i = 1}^n \int_0^1 \frac{\sum_{l = 1}^n (\hat{\Delta}_l - \Delta_l) Y_l(t)}{\sum_{l = 1}^n Y_l(t)}d\Mit \\
&=&\Big\{-\frac{1}{\sqrt{n}}\sum_{i = 1}^n \int_0^1 \frac{\sum_{l = 1}^n \tX_l^\top (g^{-1})'(\tX_l^\top \alpha_T)Y_l(t)}{\sum_{l = 1}^n Y_l(t)}d\Mit \Big\}\sqrt{n}(\hat{\alpha} - \alpha_T) + o_p(1),
\een
where $\displaystyle{\frac{\sum_{l = 1}^n \tX_l^\top(g^{-1})'(\tX_l^\top \alpha_T) Y_l(t)}{\sum_{l = 1}^n Y_l(t)}}$ will have a sup norm limit by Proposition \ref{prp:sup1}. Name that limit $\cK(t)$, we will get
\eqnn
&&\frac{1}{\sqrt{n}}\sum_{i = 1}^n \int_0^1 \frac{\sum_{l = 1}^n \tX_l^\top (g^{-1})'(\tX_l^\top \alpha_T)Y_l(t)}{\sum_{l = 1}^n Y_l(t)}d\Mit \\
&=& \frac{1}{\sqrt{n}}\sum_{i = 1}^n \int_0^1 \Big\{\frac{\sum_{l = 1}^n \tX_l^\top (g^{-1})'(\tX_l^\top \alpha_T)Y_l(t)}{\sum_{l = 1}^n Y_l(t)} - \cK(t)\Big\}d\Mit + \frac{1}{\sqrt{n}}\sum_{i = 1}^n \int_0^1 \cK(t) d\Mit.
\een
From Lemma \ref{lem:martint} and Proposition \ref{prp:sup1}, the first term will be $o_p(1)$ and second will converge weakly. Finally, rearranging the nonzero term in $\bbA_3$ gives
\eqnn
&&\rho_0\sum_{i = 1}^n (\hat{\Delta}_i - \Delta_i)(\int_0^1 \Yit dt) (\hat{\Delta}_i - \Delta_i)\\
&=& \sqrt{n}(\hat{\alpha} - \alpha_T)^\top \Big\{\frac{\rho_0}{n} \sum_{i = 1}^n (\int_0^1\Yit dt) \tX_i \tX_i^\top(g^{-1})'^2(\tX_i^\top \alpha_T)\big\} \sqrt{n}(\hat{\alpha} - \alpha_T) + o_p(1),
\een
together with the nonzero entry in $\bbA_4$,
\eqnn
\sqrt{n}(\hat{\alpha} - \alpha_T)^\top\Big\{\frac{\rho_0}{n}\sum_{i = 1}^n \int_0^1 \frac{\sum_{l = 1}^n \tX_l Y_l(t)}{\sum_{l = 1}^n Y_l(t)}\Yit \tX_i^\top (g^{-1})'^2(\tX_i^\top \alpha_T)dt\Big\} \sqrt{n}(\hat{\alpha} - \alpha_T) + o(1).
\een
The boundedness of these two terms in probability follows from similar statement.
\end{proof}

\begin{proof}[Proof of Theorem \ref{thm:regconsisor}]
With simple algebra, we have
\eqnn
U(\bbeta) &=& \frac 1n \sum_{i = 1}^n \int_0^1 (Z_i - \bar{Z}(t))(d\Nit - \Yit \bbeta^\top Z_i dt)\\
&=&\frac 1n \sum_{i = 1}^n \int_0^1 (Z_i - \bar{Z}(t))(d\Mit + \Yit d\Lambda_0(t) + \Yit\bbeta_T^\top Z_{0i} dt - \Yit\bbeta^\top Z_i dt)\\
&=&\frac 1n \sum_{i = 1}^n \int_0^1 (Z_i - \bar{Z}(t))(d\Mit + \Yit \bbeta_T^\top Z_{0i}dt - \Yit\bbeta^\top Z_i dt).
\een
Plugging in the true parameter $\bbeta_T$, we can decompose it into
\eqnn
U(\bbeta_T) &=& \frac 1n \sum_{i = 1}^n \int_0^1 (Z_i - \bar{Z}(t))d\Mit + \frac{\rho_0}{n} \sum_{i = 1}^n \big\{\int_0^1 (Z_i - \bar{Z}(t))\Yit dt\big\}(\Delta_i - \hat{\Delta}_i)\\
&=& \frac{1}{n} \sum_{i = 1}^n \int_0^1 \Big(Z_{0i}- \frac{\sone(t)}{\szero(t)}\Big)d\Mit\\
&&+ \frac{\rho_0}{n} \Big[\sum_{i = 1}^n \Big\{\int_0^1 \Big(Z_{0i} - \frac{\sone(t)}{\szero(t)}\Big)\Yit dt\Big\}\tX_i^\top(g^{-1})'(\tX_i^\top \alpha_T)\Big](\hat{\alpha} - \alpha_T)\\
&&+ \frac 1n \bbA_1 +\frac 1n \bbA_2 + \frac 1n \bbA_3 +\frac 1n \bbA_4 + o_p(1).
\een
The first term is a sample mean of $n$ i.i.d. martingale integrals of predictable functions, thus by Law of Large Numbers is $o_p(1)$. Following from Condition \ref{ass:firststep} and Law of Large Numbers again, the second term will tend to zero in probability. The remaining terms are all $o_p(1)$ by Lemma \ref{lem:regular}. Hence $U(\bbeta_T) = o_p(1)$.
%

Meanwhile, we can also express 
\eqnn
U(\bbeta_T) = U(\bbeta_T) - U(\hat{\bbeta}) = \big\{\frac{1}{n} \sum_{i = 1}^n \int_0^1 (Z_i - \bar{Z}(t))\Yit Z_i dt\big\}(\hat{\bbeta}- \bbeta_T).
\een
Solving for $\hat{\bbeta} - \bbeta_T$ we get,
\eqnn
\hat{\bbeta} - \bbeta_T &=& \Big\{\frac 1n\sum_{i = 1}^n \int_0^1 \Yit(Z_i - \bar{Z}(t))^{\otimes 2}dt\Big\}^{-1}U(\bbeta_T)\\
&=& \Big\{\frac 1n\sum_{i = 1}^n \int_0^1 \Yit\Big(Z_{0i} - \frac{\sone(t)}{\szero(t)}\Big)^{\otimes 2}dt + o_p(1)\Big\}^{-1}U(\bbeta_T).
\een
Condition \ref{ass:postivesurv} and $U(\bbeta_T) = o_p(1)$ imply the consistency of $\hat{\bbeta}$ from Slutsky's theorem. 
\end{proof}

\begin{proof}[Proof of Theorem \ref{thm:regnormal}]
Observe that
\eqnn
&&\sqrt{n}U(\bbeta_T) \\
&=& \frac{1}{\sqrt{n}} \sum_{i = 1}^n \int_0^1 \Big(Z_{0i} - \frac{\sone(t)}{\szero(t)}\Big)d\Mit\\
&&+ \frac{\rho_0}{n} \Big[\sum_{i = 1}^n \Big\{\int_0^1 \Big(Z_{0i} - \frac{\sone(t)}{\szero(t)}\Big)\Yit dt\Big\}\tX_i^\top(g^{-1})'(\tX_i^\top \alpha_T)\Big]\sqrt{n}(\hat{\alpha} - \alpha_T)\\
&&+ \frac{1}{\sqrt{n}} \bbA_1 + \frac{1}{\sqrt{n}} \bbA_2 + \frac{1}{\sqrt{n}} \bbA_3 + \frac{1}{\sqrt{n}} \bbA_4 + o_p(1)\\
&=& \frac{1}{\sqrt{n}} \sum_{i = 1}^n \int_0^1\Big (Z_{0i} - \frac{\sone(t)}{\szero(t)}\Big) d\Mit + \Psi\sqrt{n}(\hat{\alpha} - \alpha_T) + o_p(1),
\een
where 
\eqn
\Psi = \rho_0 \E \Big\{\int_0^1 \Big(Z_{01} - \frac{\sone(t)}{\szero(t)}\Big)Y_1(t) \tX_1^\top (g^{-1})'(\tX_1^\top \alpha_T)dt\Big\}.
\ee
With Condition \ref{ass:bound}, \ref{ass:firststep} and \ref{ass:baselinebound}, $\sqrt{n}U(\bbeta_T)$ can be written into a sum of i.i.d. random variables with mean zero and finite second moments. Thus the Multivariate Central Limit Theorem together with the Slutsky's Theorem proves that our estimator is asymptotically normally distributed with mean zero.

Lastly, we can compute the covariance matrix of this asymptotic normal distribution. Notice that the asymptotic covariance matrix of $\sqrt{n}U(\bbeta_T)$ is by Condtion \ref{ass:firststep},
\eqn\label{eq:abc}
&&E\Big\{\Big(\int_0^1\Big (Z_{01} - \frac{\sone(t)}{\szero(t)}\Big) dM_1(t) + \Psi \mathcal{I}^{-1}T_1\Big)^{\otimes 2}\Big\} \nonumber \\
&=& \E\Big\{\int_0^1\Big (Z_{01} - \frac{\sone(t)}{\szero(t)}\Big) dM_1(t)\Big\}^{\otimes 2} \nonumber \\\
&&+ \Psi \mathcal{I}^{-1}\E\{T_1^{\otimes 2}\}(\mathcal{I}^{-1})^\top\Psi^\top \nonumber \\\
&&+ 2\E\Big\{\int_0^1\Big (Z_{01} - \frac{\sone(t)}{\szero(t)}\Big) dM_1(t) T_1^\top\Big\}\Big](\mathcal{I}^{-1})^\top\Psi^\top \nonumber \\\
&=&\Sigma_1 + \Sigma_2 + 2\lim_{n \to \infty}\E\Big\{\sum_{i = 1}^n \int_0^1\Big (Z_{0i} - \frac{\sone(t)}{\szero(t)}\Big) d\Mit(\hat{\alpha} - \alpha_T)^\top\Big\} \Psi^\top \nonumber \\\
&=& \E\Big\{\int_0^1\Big (Z_{01} - \frac{\sone(t)}{\szero(t)}\Big)^{\otimes 2} d[M_1(t), M_1(t)]\Big\} \nonumber \\\
&&+ \Psi \mathcal{I}^{-1}\E(T_1^{\otimes 2})(\mathcal{I}^{-1})^\top\Psi^\top \nonumber \\\
&&+ 2\E\Big\{\int_0^1\Big (Z_{01} - \frac{\sone(t)}{\szero(t)}\Big) dM_1(t) T_1^\top\Big\}\Big](\mathcal{I}^{-1})^\top\Psi^\top \nonumber \\\
&=&\Sigma_1 + \Sigma_2 + 2\E\Big\{\int_0^1\Big (Z_{01} - \frac{\sone(t)}{\szero(t)}\Big) dM_1(t) T_1^\top\Big\}\Big](\mathcal{I}^{-1})^\top\Psi^\top,
\ee
where
\eqn
\Sigma_1 &=& \E\Big\{\int_0^1\Big (Z_{01} - \frac{\sone(t)}{\szero(t)}\Big)^2 d\langle M_1, M_1\rangle(t)\Big\},\\
\Sigma_2 &=& \Psi \Theta \Psi^\top,
\ee
and $\Theta$ is the variance-covariance matrix of the first step estimator.
Note that the last term in \eqref{eq:abc} is zero since it is still a martingale integral. Thus the asymptotic variance of our estimator $\sqrt{n}(\hat{\bbeta} - \bbeta) = \Omega^{-1}(\sqrt{n}U(\bbeta_T) + o_p(1))$, where 
 $\Omega$ is given in Condtion \ref{ass:postivesurv}, is clearly $\Omega^{-1}(\Sigma_1 + \Sigma_2)\Omega^{-1}$. To consistently estimate the variance, we just use their corresponding empirical parts. Specifically, $\Sigma_1$ can be instead estimated in the form of \eqref{survsigma2hat} by Lemma \ref{lem:variation}.
\end{proof}


\begin{prp}\label{prp:regbaseline}
Under \eqref{eq:survivalhazards}, \eqref{eq:errormodel}, \eqref{eq:firststep} and Conditions \ref{ass:bound},\ref{ass:asyregsurv},\ref{ass:postivesurv}, the cumulative baseline hazard function estimator defined in \eqref{regbaseline} converges in probability to the true value $\Lambda_{0T} $ of $\Lambda_{0}(\cdot) = \int_0^\cdot \bar{\lambda}_{0}(t) dt$ uniformly in $t \in [0,1]$, where $\bar{\lambda}_{0}(t) $ is the baseline hazard in equation \eqref{eq:truesurvivalmodel}, and the process $\sqrt{n}\{\hat{\Lambda}_0(\cdot) - \Lambda_{0T}(\cdot)\}$ converges weakly to a zero-mean Gaussian process whose covariance function at $(t, s)$, where $0 \le s \le t$, can be consistently estimated by
\eqn
&&n\sum_{i = 1}^n\int_0^s \frac{1}{(\sum_{j = 1}^n Y_j(u))^2}dN_i(u) + \hat{C}^\top(t) \hat{\Omega}^{-1}(\hat{\Sigma}_1 + \hat{\Sigma}_2)\hat{\Omega}^{-1} \hat{C}(s) + \hat{E}^\top(t)\hat{\Theta}\hat{E}(s)\nonumber\\
&&- \hat{C}^\top(t) \hat{\Omega}^{-1}\hat{D}(s) - \hat{C}^\top(s) \hat{\Omega}^{-1}\hat{D}(t),
\ee
where
\eqn
\hat{C}(t) &=& \int_0^t \bar{Z}(u)du,\\
\hat{D}(t) &=& \sum_{i = 1}^n\int_0^t \frac{ Z_i - \bar{Z}(u)}{\sum_{j = 1}^n Y_j(u)}dN_i(u),\\
\hat{E}(t) &=& \hat{\rho}_0\sum_{i = 1}^n\tX_i (g^{-1})'(\tX_i^\top \hat{\alpha})\int_0^t\frac{Y_i(u) }{\sum_{j = 1}^n Y_j(u)}du.
\ee
\end{prp}
\begin{proof}[Proof of Proposition \ref{prp:regbaseline}]
\eqnn
U_1(\Lambda_0(t), \bbeta, t) = \frac{1}{n}\sum_{i = 1}^n M_i(\Lambda_0(t), \bbeta, t) = \frac{1}{n}\sum_{i = 1}^n \int_0^t (dN_i(u) - Y_i(u)d\Lambda_0(u) - Y_i(t)\bbeta^\top Z_i du),
\een
note that $U_1(\hat{\Lambda}_0(t), \hat{\bbeta}, t) \equiv 0$. Thus we have
\eqnn
U_1(\Lambda_{0T}(t), \hat{\bbeta}, t) = U_1(\Lambda_{0T}(t), \hat{\bbeta}, t) - U_1(\hat{\Lambda}_0(t), \hat{\bbeta}, t) 
= \frac{1}{n} \int_0^t \sum_{i = 1}^n Y_i(u)d(\hat{\Lambda}_0(u) - \Lambda_{0T}(u)).
\een
Meanwhile, observe that 
\eqnn
&&U_1(\Lambda_{0T}(t), \hat{\bbeta}, t) \\
&=& \frac{1}{n}\sum_{i = 1}^n \int_0^t (dN_i(u) - Y_i(u)d\Lambda_{0T}(u) - Y_i(t)\hat{\bbeta}^\top Z_i du)\\
&=& \frac{1}{n}\sum_{i = 1}^n \int_0^t \{dM_i(u) - Y_i(u)(\hat{\bbeta}^\top Z_i - \bbeta^\top_T Z_{0i}) du\}\\
&=&\frac{1}{n}\sum_{i = 1}^n \int_0^t \big\{dM_i(u) - Y_i(u)(\hat{\bbeta} - \bbeta_T)^\top Z_{0i} du - Y_i(u)\hat{\rho}_0(\hat{\Delta}_i - \Delta_i)du\big\}\\
&=&\frac{1}{n}\sum_{i = 1}^n \int_0^t \big\{dM_i(u) - Y_i(u)(\hat{\bbeta} - \bbeta_T)^\top Z_{0i} du \\
&&+ Y_i(u)\hat{\rho}_0 \tX_i^\top (g^{-1})'(\tX_i^\top \alpha_T)(\hat{\alpha} - \alpha_T)du\big\} + o_p(1).
\een
where we shall emphasize that the $o_p(1)$ here is uniformly over $t \in [0,1]$, which can be easily proved with the help of Lemma \ref{lem:martint}. These together gives us that
\eqnn
&&\hat{\Lambda}_0(t) - \Lambda_{0T}(t) \\
&=& \sum_{i = 1}^n \int_0^t \frac{1}{\sum_{j = 1}^n Y_j(u)}dM_i(u) - \Big\{\sum_{i = 1}^n \int_0^t \frac{Y_i(u) Z_i^\top}{\sum_{j = 1}^n Y_j(u)} du\Big\} (\hat{\bbeta} - \bbeta_T) \\
&&- \Big\{\hat{\rho}_0\sum_{i = 1}^n \int_0^t \frac{Y_i(u) \tX_i^\top(g^{-1})'(\tX_i^\top \alpha_T) }{\sum_{j = 1}^n Y_j(u)}du\Big\}(\hat{\alpha} - \alpha_T).
\een
It is easy to establish the uniform convergence in time $t$ by checking for each term. For the weak convergence, observe that
\eqnn
&&\sqrt{n}(\hat{\Lambda}_0(t) - \Lambda_{0T}(t)) \\
&=& \sqrt{n}\sum_{i = 1}^n \int_0^t \frac{dM_i(u) - Y_i(u)(\hat{\bbeta} - \bbeta_T)^\top Z_i du- Y_i(u)\hat{\rho}_0 \tX_i^\top (g^{-1})'(\tX_i^\top \alpha_T)(\hat{\alpha} - \alpha_T)du}{\sum_{j = 1}^n Y_j(u)}\\
&=&\frac{1}{\sqrt{n}}\sum_{i = 1}^n \int_0^t \frac{1}{s^{(0)}(u)}dM_i(u) - \Big(\int_0^t \frac{s^{(1)}(u)}{s^{(0)}(u)}du\Big)^\top \sqrt{n}(\hat{\bbeta} - \bbeta_T) \\
&&- \Big(\int_0^t \frac{\gamma(u)}{s^{(0)}(u)}du\Big)^\top \sqrt{n}(\hat{\alpha} - \alpha_T) + o_p(1),
\een
where $\gamma(u) = \lim_{n \to \infty}\frac{\hat{\rho}_0}{n} \sum_{i = 1}^nY_i(u) \tX_i^\top (g^{-1})'(\tX_i^\top \alpha_T)$. With Condition \ref{ass:asyregsurv}, Martingale Central Limit Theorem 5.1.1 in \citet{fleming2011counting} can be employed to show that $\sqrt{n}(\hat{\Lambda}_0(t) - \Lambda_{0T}(t))$ converges a mean zero Gaussian process with respect to the Skorohod topology with covariance function at $0 \le s \le t$ obtained with the help of pointwise Multivariate Central Limit Theorem, 
\eqnn
&&\E\Big[\Big\{\int_0^t \frac{1}{s^{(0)}(u)}dM_1(u) - \Big(\int_0^t \frac{s^{(1)}(u)}{s^{(0)}(u)}du\Big)^\top \Omega^{-1}\int_0^1\Big (Z_{01} - \frac{\sone(u)}{\szero(u)}\Big) dM_1(u) \\
&& - \Big(\int_0^t \frac{\gamma(u)}{s^{(0)}(u)}du\Big)^\top \Psi \mathcal{I}^{-1} T_1\Big\}\Big\{\int_0^s \frac{1}{s^{(0)}(u)}dM_1(u) \\
&&- \Big(\int_0^s \frac{s^{(1)}(u)}{s^{(0)}(u)}du\Big)^\top \Omega^{-1}\int_0^1\Big (Z_{01} - \frac{\sone(t)}{\szero(t)}\Big) dM_1(t) - \Big(\int_0^s \frac{\gamma(u)}{s^{(0)}(u)}du\Big)^\top \Psi \mathcal{I}^{-1} T_1\Big\}\Big]\\
&=&\E\Big\{\int_0^s \frac{1}{(s^{(0)}(u))^2}d\langle M_1, M_1\rangle(u)\Big\} \\
&&+ \Big(\int_0^t \frac{s^{(1)}(u)}{s^{(0)}(u)}du\Big)^\top \Omega^{-1}\E\Big\{\int_0^1\Big (Z_{01} - \frac{\sone(u)}{\szero(u)}\Big)^{\otimes 2} d\langle M_1, M_1\rangle(u)\Big\}\Omega^{-1}\Big(\int_0^s \frac{s^{(1)}(u)}{s^{(0)}(u)}du\Big)\\
&&+ \Big(\int_0^t \frac{\gamma(u)}{s^{(0)}(u)}du\Big)^\top \Psi \E(T_1^{\otimes 2}) \Psi^\top \Big(\int_0^s \frac{\gamma(u)}{s^{(0)}(u)}du\Big)\\
&&- \E\Big\{\Big(\int_0^t \frac{s^{(1)}(u)}{s^{(0)}(u)}du\Big)^\top \Omega^{-1}\int_0^t\frac{Z_{01} - \frac{\sone(u)}{\szero(u)}}{\szero(u)} d\langle M_1, M_1\rangle(u)\Big\} \\
&&- \E\Big\{\Big(\int_0^s \frac{s^{(1)}(u)}{s^{(0)}(u)}du\Big)^\top \Omega^{-1}\int_0^t\frac{Z_{01} - \frac{\sone(u)}{\szero(u)}}{\szero(u)} d\langle M_1, M_1\rangle(u)\Big\} \\
&&- 2\E\Big\{\int_0^t \frac{1}{s^{(0)}(u)}dM_1(u)T_1^\top\Big\} (\mathcal{I}^{-1})^\top \Psi^\top\Big(\int_0^s \frac{\gamma(u)}{s^{(0)}(u)}du\Big) \\
&&+ 2\E\Big\{\Big(\int_0^t \frac{s^{(1)}(u)}{s^{(0)}(u)}du\Big)^\top \Omega^{-1}\int_0^1\Big (Z_{01} - \frac{\sone(u)}{\szero(u)}\Big) dM_1(u) T_1^\top\Big\} (\mathcal{I}^{-1})^\top \Psi^\top\Big(\int_0^s \frac{\gamma(u)}{s^{(0)}(u)}du\Big)
\een
where the last two terms can be proved to be zero by the similar approach as in the proof of Theorem \ref{thm:regnormal}. The remaining terms can be estimated by their empirical parts with the help of Lemma \ref{lem:variation}.
\end{proof}

\begin{proof}[Proof of Theorem \ref{thm:regsurvival}]
The proof is similarly to the proof of Theorem \ref{thm:compsurvival}, so omitted.
\end{proof}

\subsection{Competing risks}
\begin{prp}\label{prp:compcausal}
Assuming \eqref{eq:comphazards} and \eqref{eq:errormodel}, we have
\eqn
\bar{\lambda}_{10}(t|\xe,\xi,\xo) = \bar{\lambda}_{10}(t) + \be\xe + \bar{\bo}^\top\xo + \rho_0\Delta.
\ee
where $\bar{\lambda}_{10}(t) = \lambda_{10}(t) - \frac{\partial}{\partial t}\log\big[\E\big\{\exp\big(-\epsilon t \big)\big\}\big]$.
\end{prp}
\begin{proof}
The proof is completely parallel to that of Proposition \ref{prp:survivalcausal} with $1-F_1$ in place of $S$.
\end{proof}
In the following we give the regularity conditions and additional notation for  the Theorems in Section \ref{sec:competing}. 

Define 
$S^{(j)}(t) = \sum_{i = 1}^n Y_i(t)\hwit Z_i^{\otimes j} / n $ for $j = 0$ and 1. Then there exists a scalar and vector function $\szero(t)$ and $\sone(t)$, respectively, defined on $[0,1]$ such that for $j = 0,1$,
\eqn
\sup_{t \in [0,1]}||S^{(j)}(t) - s^{(j)}(t)|| \xrightarrow{\cP} 0.
\ee
The above is guaranteed by Gilvenko-Cantelli theorem, Condition \ref{ass:firststep} and  consistency of the Kaplan-Meier estimator.
\begin{ass}[Asymptotic regularity conditions]\label{ass:asyregcomp}
$\szero(t)$ is bounded away from zero.
\end{ass}

\begin{ass}\label{ass:postivecomp}
There exists a positive definite matrix $\Omega$ such that 
\eqn
\frac 1n\sum_{i = 1}^n \int_0^1 \twit\Yit\Big(Z_{0i} - \frac{\sone(t)}{\szero(t)}\Big)^{\otimes 2}dt \xrightarrow{a.s.} \boldsymbol{\Omega}.
\ee
\end{ass}
Although $w(t)$ is not adaptable to $\cF^1_{t-}$, we still have the following result:
\begin{lem}\label{lem:notmart}
For any $H(T,t)$  adaptable to $\cF_{t-}$, we have
\eqn
\E \big\{\int_0^1 H(T, t)w(t)dM^1(t)\big\} = 0.
\ee
\end{lem}
\begin{proof}
The martingale $M^1(t)$ associated with the counting process $N(t)$ has bounded variation. So we can define the integral pathwisely, i.e., 
\eqn
\int_0^1 H(T, t)w(t)dM^1(t) = \lim_{\max\{t_i - t_{i - 1}\} \to 0}\sum_{i = 1}^n H(T, t_i)w(t_i)(M^1(t_i) - M^1(t_{i - 1}))
\ee
where $\{0 = t_0 < t_1 <...< t_{n - 1} < t_n = 1\}$ is a partition on $[0,1]$. Notice that
\eqnn
\E \big\{H(T, t_i)w(t_i)(M^1(t_i) - M^1(t_{i - 1})) \big\} = \E \big\{H(T, t_i)\E\big(w(t_i)(M^1(t_i) - M^1(t_{i - 1}))\big|\cF^1_{t_{i - 1}}\big) \big\},
\een
where
\eqnn
&&\E\big\{w(t_i)(M^1(t_i) - M^1(t_{i - 1}))\big|\cF^1_{t_{i - 1}}\big\}\\
&=&\E\big\{(M^1(t_i) - M^1(t_{i - 1}))\E\big(w(t_i)\big|\cF^1_{t_i}\big)\big|\cF^1_{t_{i - 1}}\big\}\\
&=&\E\big\{(M^1(t_i) - M^1(t_{i - 1}))\E\big(\mathbbm{1}_{\{C \ge T \wedge t_i\}}\frac{G(t_i)}{G(X \wedge t_i)}\big|\cF_{t_i}\big)\big|\cF^1_{t_{i - 1}}\big\}\\
&=&\E\big\{(M^1(t_i) - M^1(t_{i - 1}))G(t_i)\big|\cF^1_{t_{i - 1}}\big\} = 0.
\een
Thus it suffices to prove the result by dominated convergence theorem, 
\eqnn
\E \big\{\int_0^1 H(T, t)w(t)dM^1(t) \big\}
&=& \E \big\{ \lim_{\max\{t_i - t_{i - 1}\} \to 0}\sum_{i = 1}^n H(T, t_i)w(t_i)(M^1(t_i) - M^1(t_{i - 1}))\big\}\\
&=& \lim_{\max\{t_i - t_{i - 1}\} \to 0}\E \big\{\sum_{i = 1}^n H(T, t_i)w(t_i)(M^1(t_i) - M^1(t_{i - 1}))\big\} \\
&=& 0.
\een
\end{proof}

Define
\eqn
\bbA_1 &=& \frac{1}{\sqrt{n}}\sum_{i = 1}^n \int_0^1 (Z_i - Z_{0i})(\hwit - w_i(t)) d\Mi1t),\\
\bbA_2 &=& \frac{1}{\sqrt{n}}\sum_{i = 1}^n \int_0^1 (Z_i - Z_{0i}) w_i(t) d\Mi1t,\\
\bbA_3 &=& \frac{1}{\sqrt{n}}\sum_{i = 1}^n \int_0^1 (\bar{Z}(t) -\bar{Z}_0(t))(\hwit - w_i(t)) d\Mi1t,\\
\bbA_4 &=& \frac{1}{\sqrt{n}}\sum_{i = 1}^n \int_0^1 (\bar{Z}(t) -\bar{Z}_0(t)) w_i(t) d\Mi1t,\\
\bbA_5 &=& \frac{\rho_0}{\sqrt{n}}\sum_{i = 1}^n \int_0^1 \Big\{\Big(Z_i - \bar{Z}(t)\Big) - \Big(Z_{0i} - \frac{\sone(t)}{\szero(t)}\Big) \Big\} \hwit\Yit (\Delta_i - \hat{\Delta}_i)dt,
\ee

\begin{lem}\label{lem:comppre}
$\bbA_1$, $\bbA_2$, $\bbA_3$, $\bbA_4$, $\bbA_5$ all converge to 0 in probability as $n\rightarrow\infty$.
\end{lem}

\begin{proof}
For the first term $\bbA_1$, by the same trick as in \citet{fine1999proportional} on the martingale expression of Kaplan-Meier estimator we have,
\eqnn
&& \frac{1}{\sqrt{n}}\sum_{i = 1}^n \int_0^1 (Z_i - Z_{0i}) (\hwit - w_i(t)) d\Mi1t\\
&=&\frac{1}{\sqrt{n}}\sum_{i = 1}^n \int_0^1 (Z_i - Z_{0i}) \rit \frac{\Gt \mathbbm{1}_{\{T^*_i < t\}}}{G(X_i \wedge t)}\Big(\sum_{j = 1}^n\int_{T^*_i}^t\frac{1}{\sum_{k = 1}^n \mathbbm{1}_{\{T^*_k > u\}}}dM_j^c(u)\Big)d\Mi1t + o_p(1)\\
&=&\frac{1}{\sqrt{n}} \sum_{j = 1}^n \int_0^1 \frac{\sum_{i = 1}^n(Z_i - Z_{0i})w_i(t) d\Mi1t \mathbbm{1}_{\{ T^*_i\leq u \leq t\}}}{\sum_{k = 1}^n \mathbbm{1}_{\{T^*_k > u\}}} dM_j^c(u) + o_p(1),
\een
where $M_j^c (u)$ is a martingale with respect to the censoring filtration
\eqn
\cF^c(u) = \{\mathbbm{1}_{\{X_i \ge t\}} , \mathbbm{1}_{\{X_i \le t, \delta_i = 0\}},\xii, \xoi, J_i, \xei, t \le u, i = 1,...,n\}.
\ee
Because the integrands are adaptable to this filtration after a careful analysis, the first term now becomes a sum of martingale integrals, which by Lemma \ref{lem:martint} is $o_p(1)$. 

Next, the second term $\bbA_2$ becomes,
\eqnn
&&\frac{1}{\sqrt{n}}\sum_{i = 1}^n \int_0^1 \{\hat{\Delta}_i - \Delta_i\}\hwit d\Mi1t\\
&=& \Big\{\frac{1}{n} \sum_{i = 1}^n \int_0^1 \tX_i^\top(g^{-1})'(\tX_i^\top \alpha_T) w_i(t) d\Mi1t\Big\} \sqrt{n}(\hat{\alpha} - \alpha_T) + o_p(1),
\een
by Lemma \ref{lem:notmart} and law of large numbers, it is $o_p(1)$.

Proofs for $\bbA_3$ and $\bbA_4$ are exactly parallel to the first two but more laborious, which we omit here.

For $\bbA_5$, we have
\eqnn
&&\frac{\rho_0}{\sqrt{n}}\sum_{i = 1}^n \int_0^1 \Big\{\Big(Z_i - \bar{Z}(t)\Big)\hwit - \Big(Z_{0i} - \frac{\sone(t)}{\szero(t)}\Big)\hwit \Big\} \Yit (\Delta_i - \hat{\Delta}_i)dt \\
&=& \Big[\frac{\rho_0}{n}\sum_{i = 1}^n \int_0^1 \Big\{\Big(Z_i - \bar{Z}(t)\Big)\hwit - \Big(Z_{0i} - \frac{\sone(t)}{\szero(t)}\Big)\hwit \Big\} \\
&&\Yit \tX_i^\top (g^{-1})'(\tX_i^\top \alpha_T)dt \Big]\sqrt{n}(\hat{\alpha} - \alpha_T),
\een
which will tend to zero in probability by law of large numbers and Gilvenko-Cantelli theorem.
\end{proof}

\begin{proof}[Proof of Theorem \ref{thm:compcons}]
By lemma \ref{lem:comppre}, we can rewrite 
$U(\bbeta_T)$ as
\eqnn
&&\frac 1n \sum_{i = 1}^n \int_0^1 \Big(Z_{0i} - \frac{\sone(t)}{\szero(t)}\Big)w_i(t) d\Mi1t + \frac 1n \sum_{i = 1}^n \int_0^1 \Big(Z_{0i} - \frac{\sone(t)}{\szero(t)}\Big)(\hwit - w_i(t)) d\Mi1t \\
&&+ \frac{\rho_0}{n} \sum_{i = 1}^n \int_0^1 \Big(Z_{0i} - \frac{\sone(t)}{\szero(t)}\Big)w_i(t)\Yit (\Delta_i - \hat{\Delta}_i)dt + o_p(1).
\een
Notice that the first term is not a martingale integral since $w_i(t)$ is not adapted to the complete data filtration $\cF^1(t)$. So partition it into
\eqnn
&&\frac 1n \sum_{i = 1}^n \int_0^1 \Big(Z_{0i} - \frac{\sone(t)}{\szero(t)}\Big)w_i(t) d\Mi1t \\
&=& \frac 1n \sum_{i = 1}^n \int_0^1 \Big(Z_{0i} - \frac{\sone(t)}{\szero(t)}\Big)\mathbbm{1}_{\{C_i \geq t\}} d\Mi1t \\
&&+ \frac 1n \sum_{i = 1}^n \int_0^1 \Big(Z_{0i} - \frac{\sone(t)}{\szero(t)}\Big)\big(w_i(t) - \mathbbm{1}_{\{C_i \geq u\}}\big) d\Mi1t.
\een
The first term tends to zero from the conclusion in censoring complete situation. For the second term, by Fubini-Toneli Theorem and the property of conditional expectation, we have
\eqnn
& &\E \Big\{\int_0^1\Big(Z_{0i} - \frac{\sone(t)}{\szero(t)}\Big)\big(w_i(t) - \mathbbm{1}_{\{C_i \geq t\}}\big) d\Mi1t \Big\} \\
&=& \E \Big[\E\Big\{\int_0^1\Big(Z_{0i} - \frac{\sone(t)}{\szero(t)}\Big)\big(w_i(t) - \mathbbm{1}_{\{C_i \geq t\}}\big) d\Mi1t | Z_i\Big\}\Big]\\
&=& \E \Big[\int_0^1\Big(Z_{0i} - \frac{\sone(t)}{\szero(t)}\Big)\E\Big\{\big(w_i(t) - \mathbbm{1}_{\{C_i \geq t\}}\big) d\Mi1t | Z_i\Big\}\Big],
\een
where 
\eqnn
\E\big\{w_1(t) dM_1^1(t)|Z_1\big\} &=& \E\big\{\mathbbm{1}_{\{C_1 \geq T_1 \wedge t\}}\Gt/G(T_1^* \wedge t) dM_1^1(t) | Z_1\big\}\\
&=& \E\big\{G(t)dM_1^1(t)\E\Big[\mathbbm{1}_{\{C_1 \geq T_1 \wedge t\}}/G(T_1^* \wedge t) |Z_{0`}, T_`, J_i\Big]|Z_1\big\}\\
&=& \E\big\{G(t)dM_1^1(t)| Z_1\big\}\\
&=& \E\big\{\E\Big[\mathbbm{1}_{\{C_1 \geq t\}}|Z_{01}, T_1, J_1\Big]dM_1^1(t) | Z_1\big\}\\
&=& \E\big\{\mathbbm{1}_{\{C_1 \geq t\}}dM_1^1(t) | Z_{01}\big\}.
\een
Hence the second term also has zero expectation. This established the consistency of the first term by Strong Law of Large Numbers. The remaining two terms can be handled by a similar way as we do in the Lemma \ref{lem:comppre}.
\end{proof}

\begin{proof}[Proof of Theorem \ref{thm:compnormal}]
Observe that
\eqnn
\sqrt{n}U(\bbeta_T) 
&=& \frac{1}{\sqrt{n}} \sum_{i = 1}^n \int_0^1\Big(Z_{0i} - \frac{\sone(t)}{\szero(t)}\Big)w_i(t) d\Mi1t \\
&&+ \Big\{\frac{\rho_0}{n} \sum_{i = 1}^n \int_0^1 \Big(Z_{0i} - \frac{\sone(t)}{\szero(t)}\Big)w_i(t)\Yit \tX_i^\top (g^{-1})'(\tX_i^\top \alpha_T)dt \Big\}\sqrt{n}(\hat{\alpha} - \alpha_T)\\
&&+ \frac{1}{\sqrt{n}} \sum_{i = 1}^n \int_0^1 \Big(Z_{0i} - \frac{\sone(t)}{\szero(t)}\Big)(\hwit - w_i(t)) d\Mi1t + o_p(1)\\
&=& \frac{1}{\sqrt{n}} \sum_{i = 1}^n \int_0^1\Big(Z_{0i} + \frac{\sone(t)}{\szero(t)}\Big)w_i(t) d\Mi1t - \Psi\sqrt{n}(\hat{\alpha} - \alpha_T)\\
&&+ \frac{1}{\sqrt{n}} \sum_{i = 1}^n \int_0^1 \frac{q(t)}{\pi(t)} dM_i^c(t) + o_p(1),
\een
where
\eqnn
\Psi &=& \rho_0 \E\Big\{\int_0^1 \Big(Z_{01} - \frac{\sone(t)}{\szero(t)}\Big)w_1(t)Y_1(t) \tX_1^\top (g^{-1})'(\tX_1^\top \alpha_T)dt\Big\}, \\
q(t) &=& - \E\Big\{\int_0^1 \mathbbm{1}_{\{T^*_1 < t \le u\}}w_1(u)\Big(Z_{01} - \frac{\sone(u)}{\szero(u)}\Big) dM_1^1(u)\Big\},\\
\pi(t) &=& \P(T^*_1 \ge t).
\een
Therefore now $\sqrt{n}U(\bbeta_T)$ has been written into a sum of i.i.d. random variables with mean zero and finite second moments. By the Multivariate Central Limit Theorem and Slutsky's theorem, our estimator is asymptotically normal with mean zero and covariance matrix $\Omega^{-1}(\Sigma_1 + \Sigma_2 + \Sigma_3)\Omega^{-1}$, 
where the cross terms are zero by similar arguments as in the proof of Theorem \ref{thm:regnormal}; here the cross terms are  between the event martingale and the censoring martingale. Note that since we assume that the event distribution and the censoring distribution are marginally independent, then
\eqnn
&&\E\Big\{\int_0^1\Big (Z_{01} - \frac{\sone(t)}{\szero(t)}\Big)w_1(t) dM_1^1(t) \int_0^1 \frac{q(t)}{\pi(t)} dM_1^c(t)\Big\}\\
&=&\E\Big\{\int_0^1\Big(Z_{01} - \frac{\sone(t)}{\szero(t)}\Big)w_1(t) dM_1^1(t)\Big\} \E\Big\{\int_0^1 \frac{q(t)}{\pi(t)} dM_1^c(t)\Big\}= 0.
\een
\end{proof}

\begin{prp}\label{prp:compbaseline}
Under \eqref{eq:comphazards}, \eqref{eq:errormodel}, \eqref{eq:firststep} on the subdistribution hazard function and assumptions \ref{ass:bound},\ref{ass:asyregcomp},\ref{ass:postivecomp}, the baseline hazard function estimator defined in \eqref{compbaseline} converges to the baseline hazard function uniformly after introducing the residual term and the process $\sqrt{n}\{\hat{\Lambda}_{10}(\cdot) - \Lambda_{10T}(\cdot)\}$ converges weakly to a zero-mean Gaussian process whose covariance function at $(t, s)$, where $0 \le s \le t$, can be consistently estimated by
\eqn
&&\int_0^s \frac{n\sum_{i = 1}^n \hat{w}_i(u) dN_i(u)}{(\sum_{j = 1}^n \hat{w}_j(u) Y_j(u))^2} + \hat{C}^\top(t) \hat{\Omega}^{-1}(\hat{\Sigma}_1 + \hat{\Sigma}_2 + \hat{\Sigma}_3)\hat{\Omega}^{-1} \hat{C}(s) + \hat{E}^\top(t)\hat{\Theta}\hat{E}(s)\nonumber\\
&&+ n\sum_{i = 1}^n\int_0^1 \frac{\hat{q}_t(u)\hat{q}_s(u)}{\hat{\pi}^2(u)}dN_i^c(u) - \hat{C}^\top(t) \hat{\Omega}^{-1}\hat{D}(s) - \hat{C}^\top(s) \hat{\Omega}^{-1}\hat{D}(t),
\ee
where
\eqn
\hat{C}(t) = \int_0^t \bar{Z}(u) du.
\ee
\end{prp}
\begin{proof}[Proof of Proposition \ref{prp:compbaseline}]
Similarly in the proof of Proposition \ref{prp:regbaseline} in the appendix, we will turn to another score function,
\eqnn
U_1(\Lambda_0(t), \bbeta, t) = \frac{1}{n}\sum_{i = 1}^n \int_0^t \hat{w}_i(u) \{dN_i(u) - Y_i(u)d\Lambda_0(u) - Y_i(t)\bbeta^\top Z_i du\},
\een
note that $U_1(\hat{\Lambda}_0(t), \hat{\bbeta}, t) \equiv 0$. Thus we have
\eqn
U_1(\Lambda_{0T}(t), \hat{\bbeta}, t) = U_1(\Lambda_{0T}(t), \hat{\bbeta}, t) - U_1(\hat{\Lambda}_0(t), \hat{\bbeta}, t) 
= \frac{1}{n} \int_0^t \sum_{i = 1}^n \hat{w}_i(u)Y_i(u)d(\hat{\Lambda}_0(u) - \Lambda_{0T}(u)).
\ee
In the meantime, observe that 
\eqnn
&&U_1(\Lambda_{0T}(t), \hat{\bbeta}, t) \\
&=& \frac{1}{n}\sum_{i = 1}^n \int_0^t \hat{w}_i(u)\{dN_i(u) - Y_i(u)d\Lambda_{0T}(u) - Y_i(t)\hat{\bbeta}^\top Z_i du\}\\
&=& \frac{1}{n}\sum_{i = 1}^n \int_0^t \hat{w}_i(u)\{dM_i(u) - Y_i(u)(\hat{\bbeta}^\top Z_i - \bbeta^\top_T Z_{0i}) du\}\\
&=&\frac{1}{n}\sum_{i = 1}^n \int_0^t \hat{w}_i(u)\{dM_i(u) - Y_i(u)(\hat{\bbeta} - \bbeta_T)^\top Z_i du - Y_i(u)\hat{\rho}_0(\hat{\Delta}_i - \Delta_i)du\}\\
&=&\frac{1}{n}\sum_{i = 1}^n \int_0^t \hat{w}_i(u)\{dM_i(u) - Y_i(u)(\hat{\bbeta} - \bbeta_T)^\top Z_{0i} du - Y_i(u)\hat{\rho}_0 \tX_i^\top (g^{-1})'(\tX_i^\top \alpha_T)(\hat{\alpha} - \alpha_T)du\}.
\een
These together gives us that
\eqnn
&&\hat{\Lambda}_0(t) - \Lambda_{0T}(t) \\
&=& \sum_{i = 1}^n \int_0^t \frac{\hat{w}_i(u)}{\sum_{j = 1}^n \hat{w}_j(u)Y_j(u)}dM^1_i(u) - \Big\{\sum_{i = 1}^n \int_0^t \frac{\hat{w}_i(u)Y_i(u) Z_{0i}^\top}{\sum_{j = 1}^n \hat{w}_j(u)Y_j(u)} du\Big\}(\hat{\bbeta} - \bbeta_T)\\
&&- \Big\{\hat{\rho}_0\sum_{i = 1}^n \int_0^t \frac{\hat{w}_i(u)Y_i(u) \tX_i^\top (g^{-1})'(\tX_i^\top \alpha_T)}{\sum_{j = 1}^n \hat{w}_j(u)Y_j(u)}du\Big\}(\hat{\alpha} - \alpha_T).
\een
Hence the uniform convergence can be established by an application of Gilvenko-Cantelli theorem. For the weak convergence, we just
need to check the variance-covariance function of $\sqrt{n}(\hat{\Lambda}_0(t) - \Lambda_{0T}(t))$. 
First note that
\eqnn
&&\sqrt{n}(\hat{\Lambda}_0(t) - \Lambda_{0T}(t)) \\
&=& \sqrt{n}\sum_{i = 1}^n \int_0^t \frac{w_i(u)}{\sum_{j = 1}^n \hat{w}_j(u)Y_j(u)}dM^1_i(u) + \sqrt{n}\sum_{i = 1}^n \int_0^t \frac{\hat{w}_i(u) - w_i(u)}{\sum_{j = 1}^n \hat{w}_j(u)Y_j(u)}dM^1_i(u)\\
&&- \sqrt{n}\sum_{i = 1}^n \int_0^t \frac{\hat{w}_i(u)Y_i(u)(\hat{\bbeta} - \bbeta_T)^\top Z_{0i}}{\sum_{j = 1}^n \hat{w}_j(u)Y_j(u)} du\\
&&- \sqrt{n}\sum_{i = 1}^n \int_0^t \frac{ \hat{w}_i(u)Y_i(u)\hat{\rho}_0 \tX_i^\top (g^{-1})'(\tX_i^\top \alpha_T)(\hat{\alpha} - \alpha_T)}{\sum_{j = 1}^n \hat{w}_j(u)Y_j(u)}du\\
&=&\frac{1}{\sqrt{n}}\sum_{i = 1}^n \int_0^t \frac{w_i(u)}{s^{(0)}(u)}dM^1_i(u) - \frac{1}{\sqrt{n}}\sum_{i = 1}^n \int_0^t \frac{w_i(u)\sum_{j = 1}^n\int_{T^*_i}^u\frac{1}{\sum_{k = 1}^n \mathbbm{1}_{\{T^*_k > v\}}}dM_j^c(v)}{s^{(0)}(u)}dM^1_i(u)\\
&&- \sqrt{n} \Big(\int_0^t \frac{s^{(1)}(u)}{s^{(0)}(u)}du\Big)^\top (\hat{\bbeta} - \bbeta_T) - \sqrt{n} \Big(\int_0^t \frac{\gamma(u)}{s^{(0)}(u)}du\Big)^\top (\hat{\alpha} - \alpha_T) + o_p(1)\\
&=&\frac{1}{\sqrt{n}}\sum_{i = 1}^n \int_0^t \frac{w_i(u)}{s^{(0)}(u)}dM^1_i(u) \\
&&+ \frac{1}{\sqrt{n}}\sum_{j = 1}^n \int_0^1 \frac{1}{\sum_{k = 1}^n \mathbbm{1}_{\{T^*_k > v\}}}\int_0^t \frac{\sum_{i = 1}^n w_i(u) dM^1_i(u)\mathbbm{1}_{\{ T^*_i < v \le u\}}}{s^{(0)}(u)}dM_j^c(v) \\
&&- \Big(\int_0^t \frac{s^{(1)}(u)}{s^{(0)}(u)}du\Big)^\top \sqrt{n}(\hat{\bbeta} - \bbeta_T) - \Big(\int_0^t \frac{\gamma(u)}{s^{(0)}(u)}du\Big)^\top \sqrt{n}(\hat{\alpha} - \alpha_T) + o_p(1),
\een
where
\eqn
\gamma(t) =\rho_0\E\big\{\hat{w}_1(t)Y_i(t) \tX_1(g^{-1})'(\tX_1^\top \alpha_T)\big\}
\ee
in the above. 
With Condition \ref{ass:asyregcomp}, Martingale Central Limit Theorem 5.1.1 in \citet{fleming2011counting} can again be employed to show that $\sqrt{n}(\hat{\Lambda}_0(t) - \Lambda_{0T}(t))$ converges in the sense of Skorohod topology to a mean zero Gaussian process with covariance function to be computed
by similar way as used in Theorem \ref{thm:regnormal} and Proposition \ref{prp:regbaseline}.
\end{proof}

\begin{proof}[Proof of Theorem \ref{thm:compsurvival}]
The consistency simply follows from the results of Theorem \ref{thm:compcons}, Proposition \ref{prp:compbaseline} in the appendix and that $\exp(-x)$ is continuously differentiable. For the asymptotic process, we just resort to delta method and Donsker theorem.
\end{proof}

\clearpage
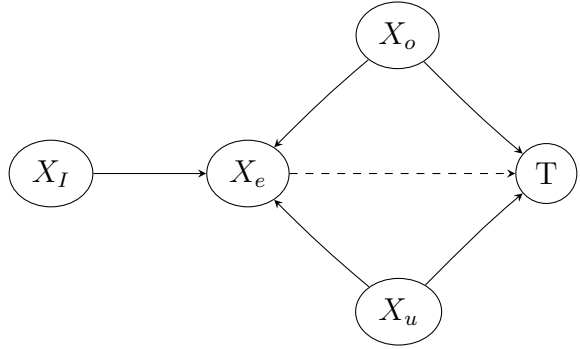
\begin{figure}[H]
\centering
	\large{\begin{tikzpicture}[%
		->,
		>=stealth,
		node distance=1cm,
		pil/.style={
			->,
			thick,
			shorten =2pt,}
		]
		\node[attribute] (1) {$X_e$};
		\node[attribute, left=1.5cm of 1] (2) {$X_I$};
		\node[attribute, right=3cm of 1] (3) {T};
    \node[attribute, above right=1.7cm of 1] (4) {$X_o$};
    \node[attribute, below right=1.7cm of 1] (5) {$X_u$};
		\draw [dashed,->] (1.east) -- (3.west);
    \draw [->] (2.east) -- (1.west);
    \draw [->] (4) to [out=220, in=45] (1);
    \draw [->] (4) to [out=315, in=140] (3);
		\draw [->] (5) to [out=140, in=315] (1);
    \draw [->] (5) to [out=45, in=220] (3);
	\end{tikzpicture}}

\caption{A causal directed acyclic graph (DAG) describing the causal relation between variables }\label{fig:causal}
\end{figure}

\begin{table}[H]
\centering
\caption{Simulation results for general survival data without competing risks.}\label{table:surv}
$$\begin{tabular}{crrrrr}\hline
 Scenario & Sample size & Bias & Emp.~Var & Est.~Var & Coverage \\ 
 \hline
I & 100 & -0.15 & 6.10 & 6.80 & 95.8\% \\ 
 I & 200 & 0.03 & 2.70 & 2.80 & 94.9\% \\ 
 I & 400 & -0.04 & 0.95 & 0.92 & 96.1\% \\ 
 I & 800 & -0.01 & 0.43 & 0.43 & 94.7\% \\ 
 I & 1200 & 0.00 & 0.28 & 0.28 & 95.3\% \\ 
  \hline
II & 100 & 0.06 & 2.10 & 2.00 & 94.6\% \\ 
 II & 200 & -0.05 & 0.93 & 0.89 & 94.4\% \\ 
 II & 400 & -0.05 & 0.45 & 0.42 & 96.0\% \\ 
 II & 800 & 0.01 & 0.21 & 0.20 & 95.2\% \\ 
 II & 1200 & 0.00 & 0.14 & 0.13 & 94.4\% \\ 
  \hline
III & 100 & 0.12 & 62.50 & 64.44 & 95.7\% \\ 
 III & 200 & 0.07 & 27.50 & 26.85 & 96.0\% \\ 
 III & 400 & 0.07 & 13.74 & 12.53 & 94.9\% \\ 
 III & 800 & 0.08 & 6.04 & 6.02 & 95.5\% \\ 
 III & 1200 & -0.04 & 4.09 & 3.96 & 95.5\% \\\hline 
\end{tabular}$$
\end{table}

\begin{table}[H]
\centering
\caption{Simulation results under the competing risks model.}\label{table:comp}
$$\begin{tabular}{crrrrr}\hline
 \multicolumn{1}{c}{Scenario} & \multicolumn{1}{c}{Sample size} & \multicolumn{1}{c}{Bias} & \multicolumn{1}{c}{Emp.~Var} & \multicolumn{1}{c}{Est.~Var} & \multicolumn{1}{c}{Coverage} \\ 
 \hline
I & 100 & 0.07 & 3.60 & 3.40 & 94.8\% \\ 
 I & 200 & 0.04 & 1.60 & 1.50 & 94.5\% \\ 
 I & 400 & 0.04 & 0.77 & 0.74 & 94.0\% \\ 
 I & 800 & 0.02 & 0.36 & 0.36 & 95.0\% \\ 
 I & 1200 & 0.01 & 0.24 & 0.24 & 94.5\% \\ 
  \hline
II & 100 & 0.13 & 4.10 & 4.30 & 95.3\% \\ 
 II & 200 & 0.04 & 2.00 & 2.00 & 94.6\% \\ 
 II & 400 & -0.04 & 0.85 & 0.95 & 95.7\% \\ 
 II & 800 & 0.01 & 0.43 & 0.46 & 95.8\% \\ 
 II & 1200 & -0.01 & 0.31 & 0.31 & 95.2\% \\ 
  \hline
III & 100 & 0.23 & 134.80 & 128.90 & 95.7\% \\ 
 III & 200 & -0.16 & 53.38 & 58.34 & 95.0\% \\ 
 III & 400 & -0.12 & 26.17 & 27.14 & 95.9\% \\ 
 III & 800 & 0.10 & 13.56 & 13.87 & 94.5\% \\ 
 III & 1200 & 0.05 & 8.85 & 8.12 & 95.5\% \\\hline
\end{tabular}$$
\end{table}

\begin{table}[H]
\centering
\caption{Patient characteristics of the SEER-Medicare data set.}\label{table:analysis}
$$\begin{tabular}{lcc}\hline
 & Radical prostatectomy & Conservative management \\ 
 & $n=10977$ & $n=18829$ \\ 
  \hline
Age & & \\ 
 66-69 & 4852 (45.3\%) & 6925 (37.2\%) \\ 
 70-74 & 5859 (54.7\%) & 11694 (62.8\%) \\ 
  \hline
Marital Status & & \\ 
 Married & 7815 (73.0\%) & 12889 (69.2\%) \\ 
 Divorced & 536 (5.0\%) & 1068 (5.7\%) \\ 
 Single & 786 (7.3\%) & 1450 (7.9\%) \\ 
 Other & 1574 (14.7\%) & 3212 (17.3\%) \\ 
  \hline
Race or Ethnity & & \\ 
 Asian & 206 (1.9\%) & 302 (1.6\%) \\ 
 Black & 1022 (9.5\%) & 2495 (13.4\%) \\ 
 Hispanic & 184 (1.7\%) & 222 (1.2\%) \\ 
 White & 8973 (83.8\%) & 15047 (80.8\%) \\ 
 Other & 326 (3.0\%) & 553 (3.0\%) \\ 
  \hline
Tumor Stage & & \\ 
 T1 & 4132 (38.6\%) & 12059 (64.8\%) \\ 
 T2 & 6579 (61.4\%) & 6560 (35.2\%) \\ 
  \hline
Tumor Grade & & \\ 
 Well differentiated & 168 (1.6\%) & 140 (0.8\%) \\ 
 Moderately differentiated & 5537 (51.7\%) & 9070 (48.7\%) \\ 
 Poorly differentiated & 4793 (44.7\%) & 9153 (49.2\%) \\ 
 Undifferentiated & 17 (0.2\%) & 26 (0.1\%) \\ 
 Cell type not determined & 196 (1.8\%) & 230 (1.2\%) \\ 
  \hline
Prior Charlson comorbidity score & & \\ 
 0 & 7217 (67.4\%) & 11868 (63.7\%) \\ 
 1 & 2301 (21.5\%) & 4260 (22.9\%) \\ 
 $\ge 2$ & 1193 (11.1\%) & 2491 (13.4\%) \\ 
  \hline
Diagnosis year & & \\ 
 2001 & 345 (3.2\%) & 241 (1.3\%) \\ 
 2002 & 311 (2.9\%) & 268 (1.4\%) \\ 
 2003 & 277 (2.6\%) & 207 (1.1\%) \\ 
 2004 & 1284 (12.0\%) & 1908 (10.2\%) \\ 
 2005 & 1217 (11.4\%) & 1838 (9.9\%) \\ 
 2006 & 1334 (12.5\%) & 2252 (12.1\%) \\ 
 2007 & 1351 (12.6\%) & 2486 (13.4\%) \\ 
 2008 & 1291 (12.1\%) & 2372 (12.7\%) \\ 
 2009 & 1215 (11.3\%) & 2381 (12.8\%) \\ 
 2010 & 1020 (9.5\%) & 2263 (12.2\%) \\ 
 2011 & 1066 (10.0\%) & 2403 (12.9\%) \\ 
  \hline
\end{tabular}$$
\end{table}

\begin{table}[H]
\centering
\caption{Results of two stage residual inclusion IV analysis on overall survival with all two-way interactions. }\label{table:ivallcauseresult}
$$\begin{tabular}{lccc}
\hline
Variable Label & Hazard difference & Standard Errors & \makecell{Two sided \\P-value} \\ 
  \hline
\makecell[l]{Radical prostatectomy vs \\Conservative management} & -0.0012 & 0.00057 & 0.042 \\ 
  Residual term & 0.0010 & 0.0006 & 0.083 \\ 
  Age 70-74 vs 66-69 & 0.0006 & 0.0005 & 0.24 \\ 
  Stage T2 vs T1 & 0.0005 & 0.0005 & 0.31 \\ 
   \hline
Married vs Other & -0.0008 & 0.0007 & 0.23 \\ 
  Divorced vs Other & 0.0011 & 0.0016 & 0.51 \\ 
  Single vs Other & -0.0006 & 0.001 & 0.53 \\ 
   \hline
Asian vs Other & -0.0004 & 0.0019 & 0.85 \\ 
  Black vs Other & -0.0004 & 0.0013 & 0.75 \\ 
  Hispanic vs Other & 0.0026 & 0.0022 & 0.24 \\ 
  White vs Other & -0.0004 & 0.0011 & 0.7 \\ 
   \hline
\makecell[l]{Grade moderately differentiated vs \\Well differentiated} & 0.0007 & 0.0004 & 0.052 \\ 
  \makecell[l]{Grade poorly differentiated vs \\Well differentiated} & 0.0019 & 0.0006 & 0.0009 \\ 
  \makecell[l]{Grade undifferentiated vs \\Well differentiated} & 0.0023 & 0.001 & 0.021 \\ 
  \makecell[l]{Grade cell type not determined vs \\Well differentiated} & 0.0035 & 0.0011 & 0.0014 \\ 
   \hline
\makecell[l]{Prior Charlson comorbidity score \\ 0 vs $\ge 2$} & -0.004 & 0.0010 & $<$ 0.0001 \\ 
  \makecell[l]{Prior Charlson comorbidity score \\ 1 vs $\ge 2$} & -0.0026 & 0.0005 & $<$ 0.0001 \\ 
   \hline
2002 vs 2001 & -0.00058 & 0.0003 & 0.022 \\ 
  2003 vs 2001 & -0.00033 & 0.0003 & 0.33 \\ 
  2004 vs 2001 & -0.0008 & 0.0004 & 0.032 \\ 
  2005 vs 2001 & -0.001 & 0.0005 & 0.022 \\ 
  2006 vs 2001 & -0.0013 & 0.0005 & 0.016 \\ 
  2007 vs 2001 & -0.0013 & 0.0006 & 0.043 \\ 
  2008 vs 2001 & -0.0014 & 0.0007 & 0.041 \\ 
  2009 vs 2001 & -0.0017 & 0.0008 & 0.036 \\ 
  2010 vs 2001 & -0.0018 & 0.0009 & 0.039 \\ 
  2011 vs 2001 & -0.0018 & 0.0010 & 0.061 \\ 
   \hline
\end{tabular}$$
\end{table}

\begin{table}[H]
\centering
\caption{Results of two stage residual inclusion IV analysis on cancer specific survival with all two-way interactions.}\label{table:cancercauseresult}
$$\begin{tabular}{lccc}
 \hline
 Variable Label & Hazard difference & Standard Errors & \makecell{Two sided \\P-value} \\ 
  \hline
\makecell[l]{Radical prostatectomy vs \\Conservative management} & -0.0002 & 0.0002 & 0.4 \\ 
  Residual term & 0.0002 & 0.0002 & 0.5 \\ 
  Age 70-74 vs 66-69 & 0.0004 & 0.0002 & 0.084 \\ 
  Stage T2 vs T1 & 0.0003 & 0.0002 & 0.23 \\ 
   \hline
Married vs Other & -0.0011 & 0.0004 & 0.0054 \\ 
  Divorced vs Other & 0.000002 & 0.0009 & 1 \\ 
  Single vs Other & -0.0005 & 0.0006 & 0.41 \\ 
   \hline
Asian vs Other & -0.0012 & 0.0011 & 0.28 \\ 
  Black vs Other & 0.0001 & 0.0007 & 0.88 \\ 
  Hispanic vs Other & -0.0004 & 0.0009 & 0.63 \\ 
  White vs Other & -0.0002 & 0.0006 & 0.69 \\ 
   \hline
\makecell[l]{Grade moderately differentiated vs \\Well differentiated} & 0.0002 & 0.0002 & 0.21 \\ 
  \makecell[l]{Grade poorly differentiated vs \\Well differentiated} & 0.0008 & 0.0003 & 0.0023 \\ 
  \makecell[l]{Grade undifferentiated vs \\Well differentiated} & 0.0016 & 0.0006 & 0.013 \\ 
  \makecell[l]{Grade cell type not determined vs \\Well differentiated} & 0.0013 & 0.0005 & 0.0075 \\ 
   \hline
\makecell[l]{Prior Charlson comorbidity score \\ 0 vs $\ge 2$} & -0.0005 & 0.0004 & 0.24 \\ 
  \makecell[l]{Prior Charlson comorbidity score \\ 1 vs $\ge 2$} & -0.0003 & 0.0002 & 0.18 \\ 
   \hline
2002 vs 2001 & -0.0001 & 0.0001 & 0.26 \\ 
  2003 vs 2001 & -0.0003 & 0.0002 & 0.059 \\ 
  2004 vs 2001 & -0.0006 & 0.0002 & 0.0013 \\ 
  2005 vs 2001 & -0.0008 & 0.0002 & 0.0008 \\ 
  2006 vs 2001 & -0.0009 & 0.0003 & 0.0016 \\ 
  2007 vs 2001 & -0.0010 & 0.0003 & 0.002 \\ 
  2008 vs 2001 & -0.0011 & 0.0004 & 0.0017 \\ 
  2009 vs 2001 & -0.0012 & 0.0004 & 0.002 \\ 
  2010 vs 2001 & -0.0014 & 0.0005 & 0.0015 \\ 
  2011 vs 2001 & -0.0015 & 0.0005 & 0.002 \\ 
   \hline
\end{tabular}$$
\end{table}

\begin{figure}[H]
  \caption{Predicted overall survival (left) and cancer specific cumulative incidence (right) function for a patient with pointwise $95\%$ confidence intervals.  }
  \centering
  \includegraphics[scale = 0.4]{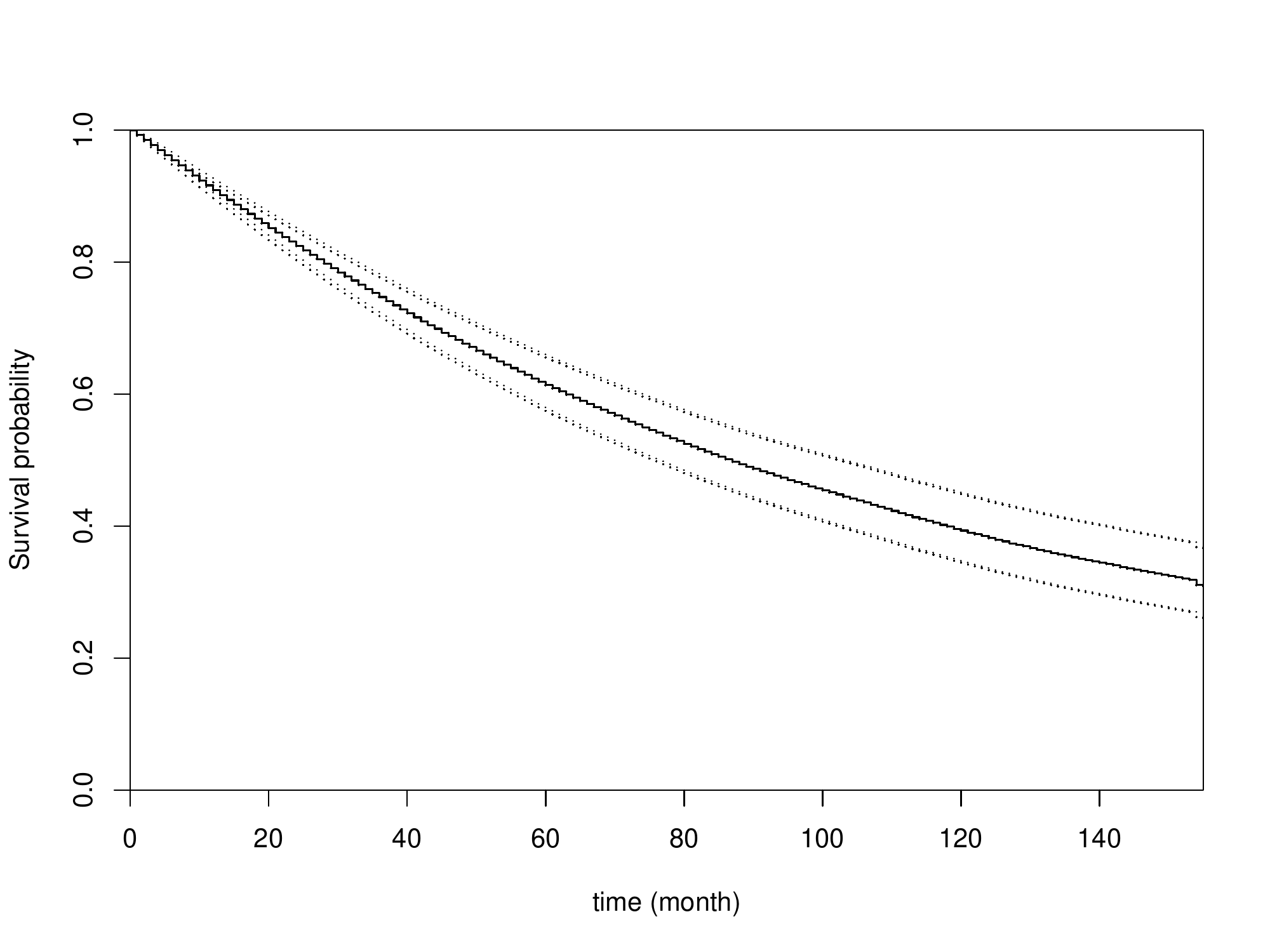}
  \includegraphics[scale = 0.4]{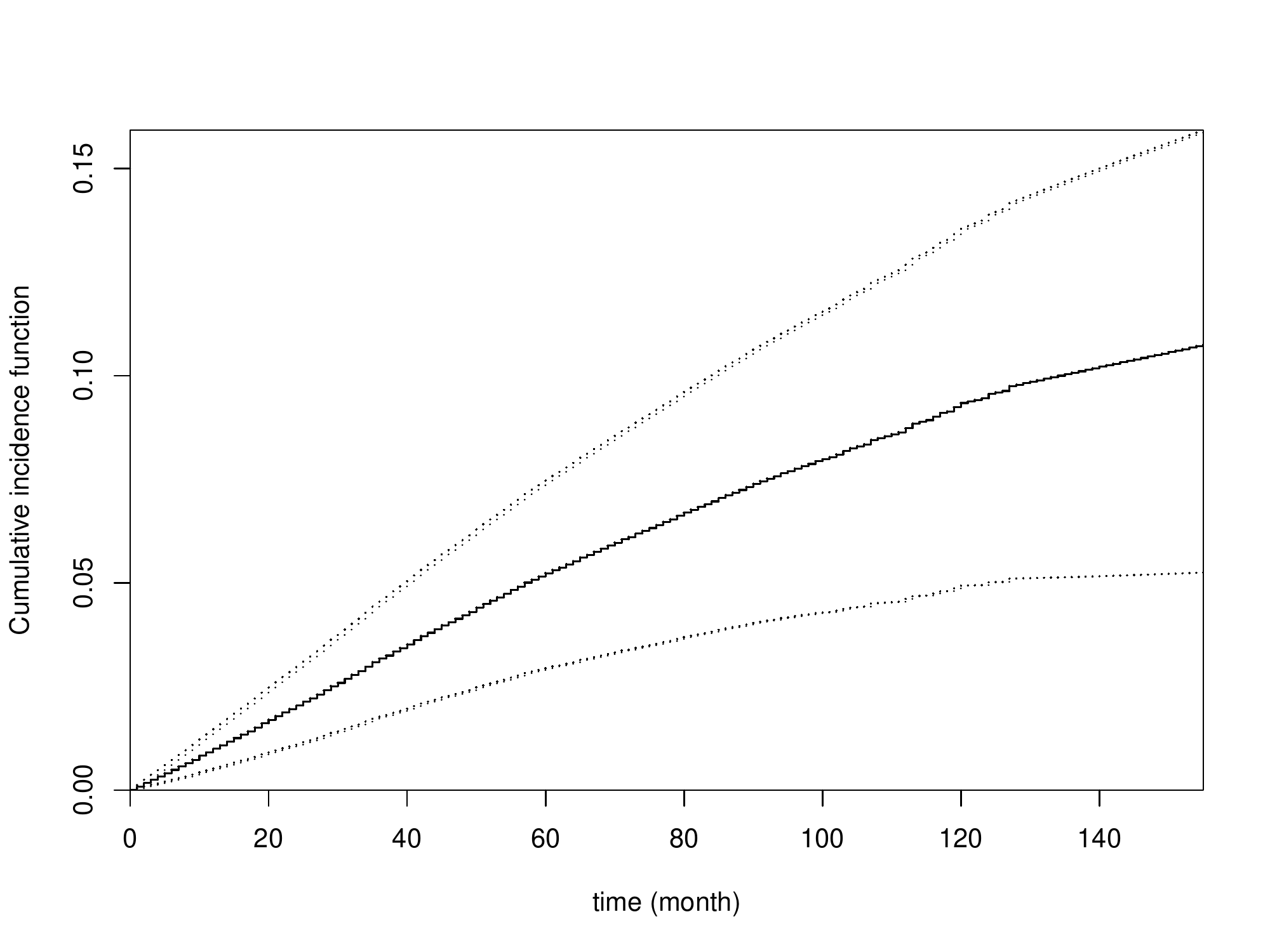}
  \label{fig:prediction}
\end{figure}
\end{document}